\documentclass[10pt,twocolumn,twoside]{IEEEtran}
\usepackage{cite}
\usepackage{amsmath,amssymb,amsfonts}
\usepackage{algorithmic}
\usepackage{graphicx}
\usepackage{textcomp}
\def\BibTeX{{\rm B\kern-.05em{\sc i\kern-.025em b}\kern-.08em
    T\kern-.1667em\lower.7ex\hbox{E}\kern-.125emX}}
\usepackage{color}
\definecolor{mygreen}{RGB}{10, 128, 128}
\definecolor{myspia}{RGB}{255, 64, 164}
\usepackage{hyperref}
\hypersetup{
	colorlinks   = true, %Colours links instead of ugly boxes
	urlcolor     = black, %Colour for external hyperlinks
	linkcolor    = mygreen, %Colour of internal links
	citecolor   = myspia %Colour of citations
}

 \usepackage{amsmath}
 \usepackage{amssymb}
\usepackage{amsthm}
\usepackage{bbm}
\usepackage{dsfont}
\usepackage{makeidx}
\usepackage{bbold}
\usepackage{latexsym,remreset}
\usepackage{multirow}
 \usepackage{amsfonts}
\usepackage{pgfplots}
\usepackage{mathtools}
\usepackage{textcomp}
\usepackage{scrextend}
\usepackage{commath}
\usepackage{graphicx}
\usepackage{url}
\usepackage{enumerate}
\usepackage{caption}
\usepackage{subcaption}
\usepackage{esvect}

\newtheorem{theorem}{Theorem}
\newtheorem{lemma}{Lemma}
\newtheorem{proposition}{Proposition}
\newtheorem{corollary}{Corollary}
\newtheorem{definition}{Definition}
\newtheorem{remark}{Remark}
\newtheorem{assumption}{Assumption}

\def \ora{\overrightarrow}

\def \G {\mathcal{G}}

\begin{document}
%
% \title{Preparation of Papers for IEEE TRANSACTIONS and JOURNALS (February 2017)}
% \author{First A. Author, \IEEEmembership{Fellow, IEEE}, Second B. Author, and Third C. Author, Jr., \IEEEmembership{Member, IEEE}
% \thanks{This paragraph of the first footnote will contain the date on 
% which you submitted your paper for review. It will also contain support 
% information, including sponsor and financial support acknowledgment. For 
% example, ``This work was supported in part by the U.S. Department of 
% Commerce under Grant BS123456.'' }
% \thanks{The next few paragraphs should contain 
% the authors' current affiliations, including current address and e-mail. For 
% example, F. A. Author is with the National Institute of Standards and 
% Technology, Boulder, CO 80305 USA (e-mail: author@boulder.nist.gov). }
% \thanks{S. B. Author, Jr., was with Rice University, Houston, TX 77005 USA. He is 
% now with the Department of Physics, Colorado State University, Fort Collins, 
% CO 80523 USA (e-mail: author@lamar.colostate.edu).}
% \thanks{T. C. Author is with 
% the Electrical Engineering Department, University of Colorado, Boulder, CO 
% 80309 USA, on leave from the National Research Institute for Metals, 
% Tsukuba, Japan (e-mail: author@nrim.go.jp).}}
%

%\title{Sufficient Conditions for the Stability of Multi-Microgrids}
%\title{A Theory for the Stability of Multi-Microgrids}
% \title{A New Certificate for Power System Stability with Applications to Multi-Microgrids}
% \title{Stability of Multi-Microgrids: A New Certificate, Impact of System Parameters, and Braess's Paradox}
\title{Stability of Multi-Microgrids: New Certificates, Distributed Control, and Braess's Paradox}
% Fast Stability Certificate in Multi-Microgrids
\author{Amin~Gholami,~\IEEEmembership{Student Member,~IEEE}~and~Xu~Andy~Sun,~\IEEEmembership{Senior Member,~IEEE}
\thanks{The authors are with the H. Milton Stewart School of Industrial and Systems Engineering, Georgia Institute of Technology, Atlanta, GA 30332 USA (e-mail: \texttt{a.gholami{@}gatech.edu}; \texttt{andy.sun{@}isye.gatech.edu}).}}

\maketitle
\begin{abstract}
This paper investigates the theory of resilience and stability in multi-microgrid (multi-$\mu$G) networks. We derive new sufficient conditions to guarantee small-signal stability of multi-$\mu$Gs in both lossless and lossy networks. The new stability certificate for lossy networks only requires local information, thus leads to a fully distributed control scheme. Moreover, we study the impact of network topology, interface parameters (virtual inertia and damping), and local measurements (voltage magnitude and reactive power) on the stability of the system. The proposed stability certificate suggests the existence of Braess’s Paradox in the stability of multi-$\mu$Gs, i.e. adding more connections between microgrids could worsen the multi-$\mu$G system stability as a whole. 
We also extend the presented analysis to structure-preserving network models, and provide a stability certificate as a function of original network parameters, instead of the Kron reduced network parameters.
We provide a detailed numerical study of the proposed certificate, the distributed control scheme, and a coordinated control approach with line switching. The simulation shows the effectiveness of the proposed stability conditions and control schemes in a four-$\mu$G network, IEEE 33-bus system, and several large-scale synthetic grids.
\end{abstract}

%\begin{IEEEkeywords}
%microgrid, electric power systems, networked control systems, nonlinear systems, stability. 
%\end{IEEEkeywords}

\IEEEpeerreviewmaketitle

\section{Introduction}
\IEEEPARstart{R}{estructuring} of distribution systems into multi-microgrids (multi-$\mu$Gs) is one of the main ways of improving the resilience of the electricity grid. The structural modularity of such networks makes them remarkably resilient against extreme events, but inherently prone to instabilities nonetheless. A minor contingency in these networks may lead to cascading outages and a total blackout in all microgrids. There is, therefore, an urgent need for understanding the notion of stability in multi-$\mu$Gs. The present paper is motivated by this urgent need and is aimed at characterizing the conditions under which a multi-$\mu$G is locally stable. We also attempt to understand how the topology and parameters of the network would affect the stability of a multi-$\mu$G, and how we can monitor and guarantee its stability using a distributed control scheme.

%As one of the basic objectives of building multi-$\mu$Gs is to keep the lights on under emergency conditions, their performance is commonly studied under such circumstances \cite{2016-Xie-multi-microgrid, 2019-gholami-sun-multi-microgrid, Vu-2019-reconfig-microgrids}. In particular, the dynamic analysis of multi-$\mu$Gs (which is the main focus of the present paper) is a fairly complex and high dimensional problem, and requires simplifying assumptions as well as an appropriate system model with a reasonable level of details.

%A key feature that distinguishes future multi-$\mu$G networks from the conventional distribution systems is that each microgrid will be connected to the rest of the system via a voltage source inverter (VSI)-based interface at the point of common coupling (PCC), and the behavior of each microgrid will be characterized by the control scheme of its VSI-based interface \cite{2016-Xie-multi-microgrid, 2019-gholami-sun-multi-microgrid}.

A key feature that distinguishes future multi-$\mu$G networks from the conventional distribution systems is that each microgrid will be connected to the rest of the system via a point of common coupling (PCC). Moreover, each microgrid either has a 
voltage source inverter (VSI)-based interface at PCC or is composed of a network of 
distributed energy resources (DERs), e.g. VSIs, diesel generators (DGs), etc
 \cite{2016-Xie-multi-microgrid, 2019-gholami-sun-multi-microgrid}.
On the other hand, it can be mathematically proved (see Lemma \ref{lemma: VSC model reparametrization}) that the frequency dynamics of a  droop-controlled VSI is equivalent to the dynamics of a synchronous generator or DG, represented by swing equations \cite{2013-schiffer-synchronization}, \cite{2010-zhong-synchronverters}.
Therefore, from a modeling perspective, the dynamical model of multi-$\mu$Gs is closely related to that of interconnected generators \cite{Vu-2019-reconfig-microgrids}, and analysis of multi-$\mu$Gs' behaviour is intertwined with an accurate understanding of swing equations.

%-------------------
Swing equations with zero transfer conductances (the so-called \emph{lossless} model) have been studied in the 1980s (see e.g. \cite{1988-Zaborszky-phase-portrait}).
%where it is commonly assumed that there is a unique stable equilibrium point (EP) and a finite number of unstable EPs in any $2\pi$ interval of generator angle coordinate.
%It is shown that the stability boundary of a stable EP consists of the stable manifolds of all the EPs (and/or closed orbits) on the stability boundary. Moreover,
%
Various methods in the broad category of the so-called direct methods have been developed to estimate the region of attraction of equilibrium points (EPs) and conduct transient stability analysis \cite{2011-Chiang-book-direct-methods, 1985-Varaiya-direct-methods}. Unfortunately, the existing direct methods are mostly cumbersome and limited to lossless systems. In a similar vein, the method in \cite{2013-anghel-algorithmic} algorithmically constructs Lyapunov functions for swing equations using semidefinite programming and sum of squares decomposition. %These Lyapunov functions can be then used for estimating the region of attraction and transient stability analysis.
Nevertheless, there are serious numerical difficulties with applying such methods to large power systems. More recently, some of these numerical problems are partially overcome by system decomposition \cite{2017-anghel-decomposition}, or by showing that there exists a convex set of Lyapunov functions certifying the transient stability of a given power system, and then finding the best suited Lyapunov function in the family to a specific contingency \cite{2016-Vu-framework}.

%------------------------------------
% LOSSY
Swing equations with nonzero transfer conductances (the so-called \emph{lossy} model) are more challenging to analyze. This is partly because there is no global energy function for such systems \cite{1989-Chiang-energy-functions-lossy}, and therefore, some basic techniques (e.g., the energy function method) for studying these equations cannot be directly applied. Nonetheless, several workarounds are developed over the years \cite{2013-anghel-algorithmic, 2017-anghel-decomposition, 1984-narasimhamurthi-lossy-energy-function, 1985-kwatny-energy-like-lyapunov, 1989-pota-lossy-lyapunov, 1979-Athay-practical-method, 2005-silva-smooth-lyapunov, 2019-josz-occupation-measures}.
For instance, the method in \cite{1989-pota-lossy-lyapunov} constructs Lyapunov functions for lossy swing equations using dissipative systems theory. However, the method still needs to be highly tuned to converge and has limitations on the bus angle differences. Athay \emph{et al.} in \cite{1979-Athay-practical-method} propose a method to estimate critical clearing times associated with first swing transient instability without explicitly solving any differential equations. The method was not supported by a sound theory at the time, but later in \cite{2005-silva-smooth-lyapunov}, a theory is also developed to support the method using the ideas of extended invariance principle.
%
% More recently, in \cite{2005-Silva-Smooth-perturbation-LossyNetworks}, the authors wrote a follow-up paper on \cite{1979-athay-practicalMethod-lossyNetworks} and showed that this approximated energy-like function is neither a Lyapunov function in the usual sense, nor an
% extended Lyapunov function, when the transfer conductances are taken into account. In spite of that, a function attending the requirements of the extension of the \textit{Invariance Principle} (the LaSalle's invariance principle in our previous discussions), that is, an extended Lyapunov function, can be obtained by smooth perturbations on that energy-like function. This perturbed function can be used to estimate the attraction area without approximations or conjectures. Indeed, the difference between the proposed extended
% Lyapunov function and the approximated energy-like function has the order of a smooth perturbation.
%

%
In \cite{2005-ortega-nontrivial-transfer-conductances2}, the lossy swing equation model is extended by considering the dynamics of the excitation system, and the asymptotic stability of EPs is ensured by proving the existence of Lyapunov functions and designing a nonlinear feedback control for the generator excitation field. The approach in \cite{2005-ortega-nontrivial-transfer-conductances2} is based on the assumption that the line transfer conductances are sufficiently small, and the angle differences at EPs are also small.
Our approach in this paper does not require these two assumptions and is aimed at finding a real-time certificate for the local stability (as opposed to global stability in \cite{2005-ortega-nontrivial-transfer-conductances2}) of EPs.

%finding an explicit relationship between local stability of the equilibrium points and network paramaters 
%
%
In \cite{1980-Skar-stability-thesis}, the local stability of swing equations with nontrivial transfer conductances is examined by linearization, and conditions for stability of EPs are established. It is also shown that undamped swing equations can be stable only under very special circumstances. Contrary to \cite{1980-Skar-stability-thesis}, the stability certificate for lossy networks in the present paper can be evaluated using only local measurements, does not require a reference bus, and reveals new relationships between network topology and stability of EPs.

Swing equations can also be studied from a graph-theoretic perspective, where the main focus is on investigating the relationship between the underlying graph structure of the power system and the system stability \cite{2015-roy-graph-theoretic, 2016-koorehdavoudi-input-output,2019-ebrahimzadeh-impact, 2018-ishizaki-graph, 2018-Dorfler-algebraic-graph-theory} . Our work in this paper falls into this research category. We refer to \cite{2018-ishizaki-graph} and \cite{2018-Dorfler-algebraic-graph-theory} for a comprehensive survey on this topic. In particular, the existing results on the small-signal stability of lossless swing equations are reviewed and studied in \cite{2018-ishizaki-graph}. It is shown that if bus angle differences at an EP are less than $\pi/2$, then the EP is locally asymptotically stable. The present paper provides a generalization of such results to lossy swing equations. Contrary to the lossless case, we will show that an EP in lossy networks could be unstable even if bus angle differences are less then $\pi/2$. 

%
%
% Another set of literature that address the stability of lossy swing equations are the recent studies of the synchronization of Kuramoto oscillators that
% are applicable to the stability analysis of lossy swing equations with strongly overdamped generators \cite{2012-Dorfler-Kuramoto}. 

Swing equations also play an important role in studying droop-controlled inverters in microgrids. 
In the literature, various models with different complexities have been adopted for droop-controlled inverters, including first-order models \cite{2013-Simpson-synchronization}, second-order models \cite{2016-Xie-multi-microgrid, 2013-schiffer-synchronization, Vu-2019-reconfig-microgrids}, third-order models \cite{2014-Schiffer-conditions}, and higher-order models  \cite{2017-vorobev-framework, 2018-Vorobev-Conference-Plug, 2019-Vorobev-Plug}.
Each model is useful for studying a particular aspect of droop-controlled inverters such as their frequency stability, voltage stability, or electromagnetic transients. The application of swing equations is more common in second-order models and frequency stability \cite{2013-schiffer-synchronization, Vu-2019-reconfig-microgrids}.
Swing equations with variable voltage magnitudes also appear in third-order models. For instance, in \cite{2014-Schiffer-conditions}, each inverter is modeled by a third-order differential equation including swing equations with variable voltage magnitudes. Using this model, sufficient conditions are derived for boundedness of trajectories in lossy microgrids as well as asymptotic stability of EPs in lossless microgrids. In the present paper, we focus on frequency stability and adopt a second-order model with constant voltage magnitudes for each inverter. In comparison with \cite{2014-Schiffer-conditions}, in the lossless case, our results match the results of \cite[See Remark 5.11]{2014-Schiffer-conditions}. In the lossy case, our sufficient condition in this paper certifies the asymptotic stability of EPs instead of boundedness of trajectories as in \cite{2014-Schiffer-conditions}. Nonetheless, our model for inverters here is different, and a direct comparison seems unfair.

The framework in \cite{2017-vorobev-framework} (and the follow-up articles \cite{2018-Vorobev-Conference-Plug, 2019-Vorobev-Plug}) utilizes a more detailed dynamical model for inverter-based microgrids, modeling the droop-based frequency and voltage controls as well as the electromagnetic transients of power lines. After performing a  model order reduction and constructing a Lyapunov function for the reduced model, a set of decentralized sufficient conditions are developed for guaranteeing the small-signal stability of the EPs. In the present paper, we pursue the same goal as in \cite{2017-vorobev-framework, 2018-Vorobev-Conference-Plug, 2019-Vorobev-Plug}, i.e., finding decentralized sufficient conditions for small-signal stability. However, our focus is more on frequency stability, and deriving more explicit stability conditions that reveal the role of network topology and parameters in small-signal stability.    

The small-signal stability of multi-$\mu$Gs is studied in \cite{2016-Xie-multi-microgrid}, where various control frameworks are proposed for the microgrids' interface. Moreover, a plug-and-play rule is proposed in \cite{Vu-2019-reconfig-microgrids}, guaranteeing the stability of multi-$\mu$Gs without requiring the global knowledge of network
topology or operating conditions.

The multi-$\mu$G model in the present paper is similar to the one in \cite{Vu-2019-reconfig-microgrids}, except we do not incorporate a local integral control. Corollary \ref{coro: paradox} in the present paper matches the plug-and-play rule in \cite{Vu-2019-reconfig-microgrids}. Moreover, the main result in Theorem \ref{thm: stability properties} generalizes the main result of \cite{Vu-2019-reconfig-microgrids} because our stability certificate considers the real-time operating condition of the system and is less conservative. Our result is also a generalization of the result in \cite{2013-schiffer-synchronization} as we do not require uniform damping of inverters.

%

%---
%-----
Another set of literature that are conceptually related to our work are the recent studies on power grid synchronization \cite{2019-paganini-global-analysis-synchronization}, frequency control \cite{2014-Low-NaLi-stability, 2018-Dorfler-robust}, voltage stability \cite{2020-ortega-tool-analysis}, and also the study of Kuramoto oscillators which has been linked to the stability analysis of lossy power systems with strongly overdamped generators \cite{2012-Dorfler-Kuramoto}.

This paper is a significant extension to our work presented in \cite{2020-fast-certificate} on a certificate for local stability of lossy swing equations. 
Compared to \cite{2020-fast-certificate}, here we further extend the certificate as a function of network physical measurements, analyze and provide a proof for the lossless case, 
%derive a transformation of the certificate for structure-preserving networks,
generalize the presented analysis to structure-preserving network models,
discuss the physical interpretation and Braess's Paradox behind it, establish a distributed control scheme based on it, and finally show its application in multi-$\mu$G networks.
The main contributions of the present paper can be summarized below.
\begin{itemize}
     \item \textbf{Stability Certificates:} We derive explicit sufficient conditions that certify small-signal stability of multi-$\mu$Gs for both lossless and lossy networks. The new certificates provide significant insights about the interplay between system stability and reactive power absorption, voltage magnitude at PCC, network topology, and interface parameters of each microgrid. We also introduce a new weighted directed graph to study the spectral properties of the multi-$\mu$G Laplacian.
     %We focus on lossy networks and discuss the lossless networks as a special case. 

     \item \textbf{Distributed Control:} In addition to providing new insights into the theory of stability, the derived stability certificates use only local information and are suitable for real-time monitoring and fast stability assessment. Based on the developed theory, we introduce a fully distributed control scheme to adjust the dynamic parameters of each microgrid interface for maintaining the stability of the system.

     \item \textbf{Analog of Braess's Paradox:} 
     %We characterize the effect of network topology, interface parameters (virtual inertia and damping), and local measurements (voltage magnitude and injected reactive power) on the system stability. 
     The stability conditions developed in this paper surprisingly reveal an analog of Braess’s Paradox in power system stability, showing that adding power lines to the system may decrease the stability margin \cite{2005Braess}.
    %  This paradox in the stability of multi-$\mu$Gs has been also observed in \cite{Vu-2019-reconfig-microgrids} based on a completely different analysis. Additionally, Braess’s paradox has analogies in the robustness of power systems \cite{2015-robustness-paradox-network}. 
     The current paper rigorously establishes the impact of switching-off lines, increasing damping, and decreasing virtual inertia on improving system stability. 
     %Moreover, the interplay between system stability, reactive power, and voltage magnitude of each microgrid PCC is illustrated.
     %-------- added in R1
     %\textcolor{blue}{
     \item \textbf{Generalization to Structure-Preserving Models:} We extend the presented analysis to structure-preserving power network models. Specifically, we prove a monotonic relationship between entries of a nodal admittance matrix and its Kron reduced counterpart. This monotonic relationship enables us to derive a stability condition as a function of original network parameters, instead of the Kron reduced network parameters. This is beneficial to real-time distributed control as the network parameters constantly change and Kron reduction may not be available to individual controllers.
     %(which is a common practice in the literature).
     %
    %  To establish this result, we show the invariance of a certain class of Laplacian matrices under Schur complementation, which is of independent interest.
     
\end{itemize}

We believe the findings in this paper are also applicable to other problems whose models display similar structural properties, such as small-signal stability assessment in the transmission level and synchronization of coupled second-order nonlinear oscillators. The rest of our paper is organized as follows. Section \ref{sec: Background} provides a brief background on multi-$\mu$Gs. In Section \ref{Sec: Linearization and Spectrum of Jacobian}, the multi-$\mu$G model is linearized
%the linkage between the linearized model and the underlying graph of the power grid is established, 
and several properties of the Jacobian matrix are proved. 
Section \ref{Sec: Stability and Hyperbolicity of the Equilibrium Points} is devoted to the main results on sufficient conditions for the stability of multi-$\mu$Gs. Section \ref{Sec: Computational Experiments} further illustrates the developed analytical results through numerical examples, and finally, the paper concludes with Section \ref{Sec: Conclusions}.
\section{Background} \label{sec: Background}
\subsection{Notations}
We use $\mathbb{C_{-}}$ to denote the set of complex numbers with negative real part, and $\mathbb{C}_{0}$ to denote the set of complex numbers with zero real part. $j=\sqrt{-1}$ is the imaginary unit. The spectrum of a matrix $A\in\mathbb{R}^{n\times n}$ is denoted by $\sigma(A)$.
For $A\in\mathbb{C}^{n\times n}$ and $\alpha,\beta \subseteq\{1,...,n\}$, the submatrix of entries in the rows indexed by $\alpha$ and
columns indexed by $\beta$ is denoted by $A[\alpha, \beta]$. Similarly, for a vector $x\in\mathbb{C}^{n}$, $x[\alpha]$ denotes the subvector consisting of entries indexed by $\alpha$.

\subsection{Multi-Microgrid Model} \label{subsec: Multi-Machine Swing Equations}
Consider a distribution network represented as an undirected graph $\mathcal{G}=(\mathcal{N},\mathcal{E})$, where $\mathcal{N}$ is the set of nodes and $\mathcal{E}$ is the set of edges. Each node in $\mathcal{N}$ represents a microgrid and each edge in $\mathcal{E}$ represents an electrical branch connecting the two microgrids across the branch. We will refer to $\mathcal{G}$ as the linking grid \cite{2019-gholami-sun-multi-microgrid}.
The linking grid $\mathcal{G}$ is \emph{connected} if 
     for any two nodes $i,k \in \mathcal{N}, i\ne k$ there exists a path between $i$ and $k$ consisting of power lines with nonzero admittance.
     
%\textcolor{blue}{The linking grid $\mathcal{G}$ is \emph{connected} if 
    % for any two nodes $i,k \in \mathcal{N}, i\ne k$, there exists an ordered sequence of nodes between $i$ and $k$ such that any pair of consecutive nodes in the sequence are connected by a power line $\mathcal{E}$ represented by a nonzero admittance. }
%\textcolor{blue}{
To begin the study, let us assume that each microgrid is modeled by a grid-forming VSI connected to a node of the linking grid.
%(we will discuss later in Section \ref{subsec: structure-preserving microgrid} how this assumption can be relaxed).
%
%
%
Given the time window of small-signal stability assessment, 
%each microgrid can be approximated by a VSG \cite{2016-Xie-multi-microgrid, Vu-2019-reconfig-microgrids}, the existence of diesel generators in Scenarios (b)-(c), wind generation units in Scenario (d) or other renewable or battery sources (interfaced with a power electronics inverter operated in the grid-forming mode) would be acceptable.
characterization of each microgrid by a VSI can be understood in two ways:
\begin{enumerate}
    \item The first possibility is that a microgrid contains an ensemble of devices (e.g., grid-forming inverters, diesel generators, and loads) whose aggregate behavior can be modeled by a VSI. The derivation of the aggregated VSI model is out of the scope of this paper. We refer to \cite{2019-Johnson-aggregate, 2018-Johnson-aggregate} for details. Moreover, we restrict the type of microgrid DERs to grid-forming VSIs, DGs, and more generally to those whose dynamics can be captured by swing equations.
    \label{way1}
    \item The second possibility is that a microgrid is connected to the linking grid through a grid-forming VSI at PCC \cite{2016-Xie-multi-microgrid, 2019-gholami-sun-multi-microgrid, Vu-2019-reconfig-microgrids}. VSI-based interfaces decouple the intra-microgrid dynamics from the grid side, and consequently, the interactions among different microgrids will be primarily determined by the VSI control law \cite{2016-Xie-multi-microgrid, 2019-gholami-sun-multi-microgrid}.
    %In this case, from the viewpoint of linking grid, the behaviour of microgrid can be characterized by the VSI interface.
    \label{way2}
\end{enumerate}
When the model order reduction in way \ref{way1} introduces major errors, or VSI interfaces in \ref{way2} do not exist, it is inevitable that the internal structure, DERs, and loads of the microgrid be explicitly modeled.
%each microgrid will be modeled by a higher-order dynamical model preserving .
%modeling the the network inside each microgrid seems inevitable.
Later in Section \ref{subsec: structure-preserving microgrid}, we will introduce a way to consider a structure-preserving model for each microgrid and extend our stability analysis to such cases.
%study the role of network topology inside each microgrid on the stability of the multi-microgrid.
Let us for now focus on the case where each microgrid is represented as a node in the linking grid $\mathcal{G}=(\mathcal{N},\mathcal{E})$ and modeled by a VSI.
%extend model \eqref{eq: swing equations} by considering the network structure inside each microgrid $i\in\mathcal{N}$.

 Accordingly, the dynamics of a multi-$\mu$G network is characterized by the following system of nonlinear autonomous ordinary differential equations (ODEs):
%\vspace{-4mm}
\begin{subequations} \label{eq: swing equations}
	\begin{align}
	& \dot{\delta}_i(t) = \omega_i(t) && \forall i \in \mathcal{N},  \label{eq: swing equations a}\\
	& m_i \dot{\omega}_i(t)+ d_i \omega_i(t) = P_{s_i} - P_{e_i}(\delta(t)) && \forall i \in \mathcal{N}, \label{eq: swing equations b}
	\end{align}
\end{subequations}
where for each microgrid $i\in\mathcal{N}$, $P_{s_i}$ is the active power setpoint in per unit, $P_{e_i}$ is the outgoing active power flow in per unit, $m_i$ is the virtual inertia in seconds induced by the delay in droop control, $d_i$ is the unitless damping coefficient, $t$ is the time in seconds, $\delta_i(t)$ is the terminal voltage angle in radian, and finally $\omega_i(t)$ is the deviation of the angular frequency from the nominal angular frequency in radian per seconds. For the sake of simplicity, henceforth we do not explicitly write the dependence of the state variables $\delta$ and $\omega$ on time $t$. 

%Moreover, for the sake of small-signal stability assessment following a disturbance, we assume the setpoints $P_{s_i}, \forall i \in \mathcal{N}$ are constant. This is a common assumption in studying small-signal stability with fast dynamics compared to the change of setpoints. 

The PCC of two microgrids $i$ and $k$ are connected via a power line with the admittance $y_{ik}=g_{ik}+jb_{ik}$, where $g_{ik}\ge0$ and $b_{ik}\le0$. In transmission-level small-signal stability studies, the conductance $g_{ik}$ of transmission lines is commonly assumed to be zero (aka lossless model). While this is a reasonable assumption in the transmission level, it may not hold in the distribution level and multi-$\mu$G networks. Therefore, our analysis will be based on the general lossy case, and we discuss the lossless model as a special case. Let $y_{ii}$ denote the admittance-to-ground at PCC $i$ and define the symmetric admittance matrix given by the diagonal elements $Y_{ii}\measuredangle \theta_{ii}=\sum_{k=1}^n y_{ik}$ and off-diagonal elements $Y_{ik}\measuredangle \theta_{ik}=-y_{ik}$. Based on this definition, the function $P_{e_i}$ in \eqref{eq: swing equations b} can be further spelled out: 
\begin{align} \label{eq: flow function}
	P_{e_i}(\delta) & = \sum \limits_{k = 1}^n { V_i  V_k Y_{ik} \cos \left( \theta _{ik} - \delta _i + \delta _k \right)},
\end{align}
where $V_i$ is the PCC terminal voltage magnitude of microgrid $i$.

\begin{definition}[flow function] \label{def: flow function}
The smooth function $P_e:\mathbb{R}^n \to \mathbb{R}^n$ given by $\delta \mapsto P_e(\delta)$ in \eqref{eq: flow function} is called the flow function.
\end{definition}
The smoothness of the flow function (it is $\mathcal{C}^\infty$ indeed) is a sufficient condition for the existence and uniqueness of the solution to the ODE \eqref{eq: swing equations}. 
The flow function is translationally invariant with respect to the operator $\delta \mapsto \delta + \alpha \mathbf{1}$, where $\alpha \in \mathbb{R}$ and $\mathbf{1}\in\mathbb{R}^n$ is the vector of all ones. In other words, $P_e(\delta + \alpha \mathbf{1})=P_e(\delta)$. A common way to deal with this situation is to define a reference bus and refer all other bus angles to it. This is equivalent to projecting the original state space onto a lower dimensional space. 

Observe that EPs of the multi-$\mu$G dynamical system \eqref{eq: swing equations} are of the form $(\delta^*,\omega^*)\in\mathbb{R}^{2n}$ where $\delta^*$ is a solution to the active power flow problem $P_{e_i}(\delta^*) = P_{s_i}, \forall i \in \mathcal{N}$ and $\omega^*=0$. We seek an answer to the following question: under what conditions is an EP $(\delta^*,\omega^*)$ locally asymptotically stable? A perfect answer to this question should give us a purely algebraic condition, shedding light on the relation between the stability of the EP and the parameters of system \eqref{eq: swing equations} (i.e., the interface parameters $m_i$ and $d_i$, the setpoints $P_{s_i}$, and the underlying graph of the multi-$\mu$G network). The rest of this paper is devoted to finding such an answer.

As mentioned before, model \eqref{eq: swing equations} is identical to the well-known swing equation model describing the dynamics of interconnected synchronous generators \cite{2008-anderson-stability}, and this is because VSI control schemes are widely devised to emulate the behavior of synchronous machines \cite{2016-Xie-multi-microgrid,2019-gholami-sun-multi-microgrid, 2010-zhong-synchronverters}.
Indeed, the equivalence of the dynamics of synchronous generators and droop-controlled VSIs can be rigorously formalized.
%\cite[Lemma 4.1]{2013-schiffer-synchronization}. 
%---------------------------
% Added in R1
%\textcolor{blue}{
Specifically, a droop-controlled VSI at node $i \in \mathcal{N}$ can be modeled as \cite{2013-schiffer-synchronization}:
\begin{subequations} \label{eq: VSC}
	\begin{align}
	& \dot{\delta}_i(t) = -k_i\left( P_{m_i}(t) - P_{d_i} \right) ,  \label{eq: VSC a}\\
	& \tau_i \dot{P}_{m_i}(t) = -P_{m_i}(t) + P_{e_i}, \label{eq: VSC b}
	\end{align}
\end{subequations}
where $k_i\ge0$ is the droop gain, $P_{m_i}\in\mathbb{R}$ is the measured active power, $P_{d_i}\in\mathbb{R}$ is the desired active power setpoint, and $\tau_i\ge0$ is the time constant of the low-pass filter of the power measurement.
%
% For the sake of quick reference, the swing equation model at a node $i \in \mathcal{N}$ is
% \begin{subequations} \label{eq: swing equations}
% 	\begin{align}
% 	& \dot{\delta}_i(t) = \omega_i(t),  \label{eq: swing equations a}\\
% 	& m_i \dot{\omega}_i(t)+ d_i \omega_i(t) = P_{s_i} - P_{e_i}(\delta(t)). \label{eq: swing equations b}
% 	\end{align}
% \end{subequations}
Now, the next lemma shows the droop-controlled VSI model \eqref{eq: VSC} can be reparametrized as the swing equation model \eqref{eq: swing equations}.
%the equivalence of the dynamics of the droop-controlled VSC model \eqref{eq: VSC} and the swing equation model \eqref{eq: swing equations}.
%
\begin{lemma}[VSI model reparametrization] \label{lemma: VSC model reparametrization}
The dynamics of the droop-controlled VSI model \eqref{eq: VSC} is equivalent to the dynamics of the swing equation model \eqref{eq: swing equations}.
\end{lemma}
\begin{proof}
Consider the VSI model \eqref{eq: VSC} and define the new variable $\omega_i(t)$ as
\begin{align}
 \omega_i(t):= \dot{\delta}_i(t) = -k_i\left( P_{m_i}(t) - P_{d_i} \right).  
\end{align}
Thus, using the new variable $\omega_i(t)$, equation \eqref{eq: VSC a} can be written as \eqref{eq: swing equations a}. By substituting $P_{m_i}(t) = -\omega_i(t)/k_i + P_{d_i} $ into \eqref{eq: VSC b}, we get
\begin{align*}
    - \tau_i \dot{\omega}_i(t)/k_i  = \omega_i(t)/k_i - P_{d_i} + P_{e_i}.
\end{align*}
Now, for each node $i \in \mathcal{N}$, define the virtual inertia coefficient $m_i:={\tau_i}/{k_i}$, virtual damping $d_i:=1/k_i$, and active power setpoint $P_{s_i}:=P_{d_i}$. Therefore, \eqref{eq: VSC b} is equivalent to \eqref{eq: swing equations b}.
\end{proof}
Similar derivations can be found for example in \cite{2013-schiffer-synchronization, 2020-Guerrero-VSG}.
%}
%
%---------------------------------------
%\textcolor{blue}{
Accordingly, the application of model \eqref{eq: swing equations} is not restricted to the characterization of interconnected microgrids. The model (and consequently, the results developed in this paper) can be applied to a system of interconnected synchronous machines, coupled oscillators, etc.
%}
\section{Linearization and Spectrum of Jacobian} \label{Sec: Linearization and Spectrum of Jacobian}
\subsection{Linearization}
Let us take the state variable vector $(\delta,\omega)\in\mathbb{R}^{2n}$ into account and note that the first step in studying the stability of multi-$\mu$G EPs is to analyze the Jacobian of the vector field in \eqref{eq: swing equations}: 
\begin{align}\label{eq: J}
J := \begin{bmatrix}
0 & I \\
-M^{-1} L  & - M^{-1}D \\
\end{bmatrix} \in\mathbb{R}^{2n\times 2n}
\end{align}
where $I\in\mathbb{R}^{n\times n}$ is the identity matrix, $M=   \mathbf{diag}(m_1,\cdots,m_n)$, and $D=\mathbf{diag}(d_1,\cdots,d_n)$. Throughout the paper, we assume $M$ and $D$ are nonsingular. Moreover, $L\in\mathbb{R}^{n\times n}$ is the Jacobian of the flow function with the entries:
%\vspace{-5mm}
\begin{subequations}\label{eq: Jacobian}
\begin{align} 
& L_{ii}= \sum \limits_{k=1, k \ne i}^n { V_i V_k Y_{ik} \sin \left( {\theta _{ik} - {\delta _i} + {\delta _k}} \right) }, \forall i \in \mathcal{N} \label{eq: Jacobian1}\\
&L_{ik}=  - {V_i} {V_k}{Y_{ik}}\sin \left( {{\theta _{ik}} - {\delta _i} + {\delta _k}} \right),\forall i\neq k \in \mathcal{N} \label{eq: Jacobian2}.
\end{align}
\end{subequations}
The matrix $L$ plays a prominent role in the spectrum of the Jacobian matrix $J$ (and as a consequence, in the stability properties of the EPs of multi-$\mu$Gs). We illustrate this role in the following subsection.

\subsection{Spectral Relationship Between Matrices $J$ and $L$}
The next lemma shows that the eigenvalues of $J$ and $L$ are linked through a singularity constraint. 
%Before proceeding to investigate this relationship, let us define the concept of matrix pencil \cite{2001-Tisseur-Pencil}. 
Recall that for $n\times n$ real matrices $Q_0,Q_1,$ and $Q_2$, a \emph{quadratic matrix pencil} is a matrix-valued function $P:\mathbb{C}\to\mathbb{R}^{n \times n}$ given by $\lambda \mapsto P(\lambda)$ such that  $P(\lambda) = \lambda^2Q_2 + \lambda Q_1 + Q_0$.
\begin{lemma}\label{lemma: relation between ev J and ev J11}
	$\lambda$ is an eigenvalue of $J$ if and only if the quadratic matrix pencil $P(\lambda):= \lambda^2 M + \lambda D + L$ is singular.
\end{lemma}
%
% \begin{proof}
% 	Let $\lambda$ be an eigenvalue of $J$ and $(v_1 , v_2)$ be the corresponding eigenvector. Then	
% 	\begin{align} \label{eq: J cha eq}
% 	\begin{bmatrix}
% 	0 & I \\
% 	-M^{-1} L     &     - M^{-1}D \\
% 	\end{bmatrix}   \begin{bmatrix} v_1 \\v_2  \end{bmatrix}  = \lambda   \begin{bmatrix} v_1 \\v_2  \end{bmatrix}
% 	\end{align}
% 	which implies that $ v_2 = \lambda v_1$ and $(M^{-1} L + \lambda ( M^{-1} D  + \lambda I))v_1 = 0$, and consequently, 
% 	\begin{align}
% 	& \left(   \lambda^2 M + \lambda D + L \right) v_1 = 0. \label{eq: quadratic matrix pencil}
% 	\end{align}  
% 	Since the eigenvector $v$ is nonzero, we have $v_1 \not = 0$ (otherwise $v_2 = \lambda \times 0 = 0 \implies v = 0 $). Equation \eqref{eq: quadratic matrix pencil} implies that the matrix pencil $P(\lambda)= \lambda^2 M + \lambda D + L$ is singular.

% 	Conversely, 
% 	suppose there exists $\lambda \in \mathbb{C}$ such that $P(\lambda)= \lambda^2 M + \lambda D + L$ is singular. Choose a nonzero $v_1 \in   \mathbf{ker}(P(\lambda))$ and let $ v_2 := \lambda v_1$. 
% 	Accordingly, the characteristic equation \eqref{eq: J cha eq} holds, and consequently, $\lambda$ is an eigenvalue of $J$.
% \end{proof}
% An interested reader is referred to \cite[Proposition 5.14]{2018-Dorfler-algebraic-graph-theory} for similar results under different settings.
The proof of Lemma \ref{lemma: relation between ev J and ev J11} is given in \cite{2020-fast-certificate}.
Next, Proposition \ref{prop: geometric nullity of J and J11} illustrates the relationship between the kernels and the multiplicity of the zero eigenvalue of the two matrices $J$ and $L$.
%----------
\begin{proposition} \label{prop: geometric nullity of J and J11}
Consider the Jacobian matrix $J$ in \eqref{eq: J}. The following statements hold:
\begin{enumerate}[(i)]
    \item The kernel of $L$ is the orthogonal projection of the kernel of $J$ onto the linear subspace $\mathbb{R}^n \times \{0\}$. That is, $\mathbf{ker}(L)=  \mathbf{proj} (\mathbf{ker}(J))$. 
    
    \item  The geometric multiplicity of the zero eigenvalue in $\sigma(J)$ and $\sigma(L)$ are equal.
    
    \item $J$ is nonsingular if and only if $L$ is nonsingular.
\end{enumerate}
\end{proposition}
\begin{proof}
 See Appendix \ref{proof of prop: geometric nullity of J and J11}.
\end{proof}

As the role of $L$ in the spectrum of $J$ became more clear, we scrutinize the spectrum of $L$ in the next subsection. Our final goal is to use the spectral properties of $L$ together with the relationships established in Lemma \ref{lemma: relation between ev J and ev J11} and Proposition \ref{prop: geometric nullity of J and J11} to derive a stability certificate for multi-$\mu$Gs.

\subsection{A Directed Graph Induced by $L$} \label{subsec: digraph}
The linking grid of a multi-$\mu$G is represented by the undirected graph $\G$ defined in Section \ref{subsec: Multi-Machine Swing Equations}. However, to fully represent the  Jacobian $L$ of the flow function \eqref{eq: flow function}, we need to introduce a new \emph{weighted directed} graph (digraph). 
%is intertwined with the graph structure of the multi-$\mu$G network. To see this relation, let us define a weighted directed graph (digraph) 
Let $\ora{\G}=(\mathcal{N},\mathcal{A},\mathcal{W})$, where each node $i\in\mathcal{N}$ corresponds to a microgrid and each directed arc $(i,k)\in\mathcal{A}$ corresponds to the entry $(i,k), i\ne k$ of the admittance matrix. We further define a weight for each arc  $(i,k)\in\mathcal{A}$:
\begin{align} \label{eq: weights of digraph lossy}
w_{ik} = {V_i} {V_k}{Y_{ik}}\sin \left( \varphi_{ik} \right), 
%\quad\forall (i,k)\in\mathcal{A},
\end{align}
where $\varphi_{ik} := {{\theta _{ik}} - {\delta _i} + {\delta _k}}$. With the above definitions, we can see that the Jacobian matrix $L$ of the flow function, which appeared in \eqref{eq: Jacobian}, is indeed the \emph{Laplacian} of the weighted digraph $\ora{\G}$.
%
%defined as $L = D^+(\ora{\G}) - A(\ora{\G})$, where $D^+(\ora{\G})$ is a diagonal matrix with the $i$-th diagonal entry being the sum of all the weights of the out-going arcs from node $i$, i.e., $D^{+}_{ii}(\ora{\G})=\sum_{k=1, k\ne i}^n w_{ik}$. Moreover, $A(\ora{\G})$ is the adjacency matrix of $\ora{\G}$, i.e. $A_{ik}$ is the weight of the arc $(i,k)$ and $0$ if there is no arc between $i$ and $k$.
%
In general, the arc weights $w_{ik}$ can take any values in $\mathbb{R}$, and the matrix $L$ is not necessarily symmetric. In practice, however, $w_{ik}$ varies in a small nonnegative range. Figure \ref{fig: histogram of angles} illustrates the histogram of the angle $\varphi_{ik}$ for all lines $(i,k)$ in different IEEE and NESTA standard distribution test cases, where the converged load flow data are obtained from \textsc{MATPOWER} \cite{matpower}. Accordingly, $\varphi_{ik} \in (0,\pi)$ in all of these cases. 
%This observation stems from the fact that the entries of the admittance matrix, i.e., $Y_{ik} \measuredangle \theta_{ik}$ satisfy the following inequalities: $\theta_{ii} \in [ -\frac{\pi}{2} , 0 ), \: \forall i\in\mathcal{N}$ and $ \theta_{ik} \in [\frac{\pi}{2} , \pi), \: \forall (i,k)\in\mathcal{A}$ where the lower bounds are realized in lossless networks.  
%
Thus, it is reasonable to assume that the EPs $(\delta^{*},\omega^{*})$ of the multi-$\mu$G dynamical system \eqref{eq: swing equations} are located in the set $\Omega$ defined as
\begin{align*}
\hspace{-2mm}\Omega = \left\{ (\delta,\omega)\in\mathbb{R}^{2n} :   0 < \theta_{ik}-\delta_i+\delta_k < \pi , \forall (i,k) \in \mathcal{A},\omega = 0  \right\}.
\end{align*}
Under this assumption, the arc weights $w_{ik}> 0$ for all arcs $(i,k)$. So, there are two arcs $(i,k)$ and $(k,i)$ between microgrids $i$ and $k$ if and only if the two microgrids are physically connected. We always assume the physical network connecting all the microgrids is a connected (undirected) graph. The weighted digraph $\ora{\G}$ will be used to study the spectral properties of $L$.

%------------------------------------
\begin{figure}[t]
 \includegraphics*[width=3.5in, keepaspectratio=true]{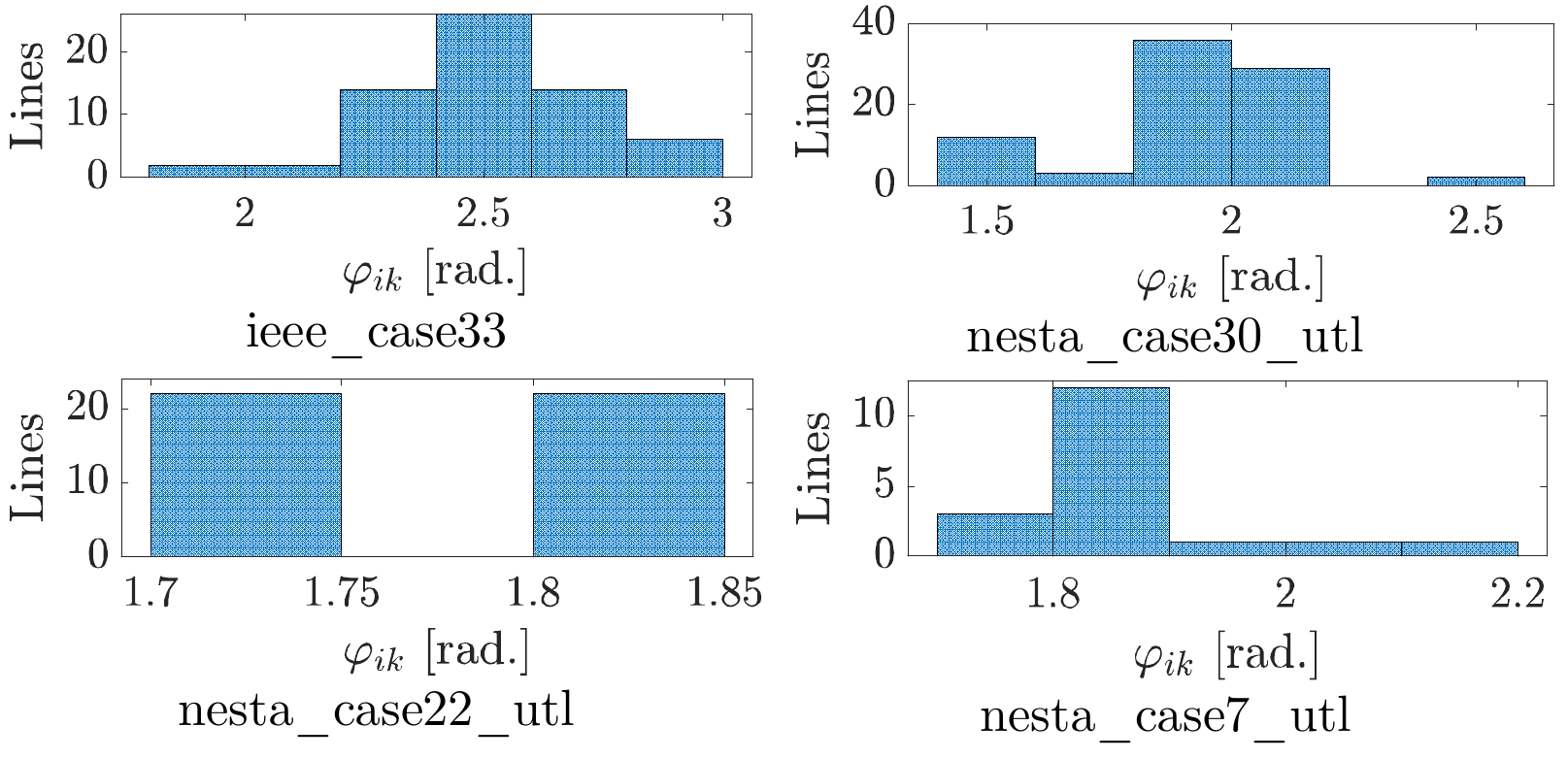}
 \centering
  \caption{Histogram of the distribution of $\varphi_{ik}$ for all lines $(i,k)$ in different IEEE and NESTA standard test cases.}
  \label{fig: histogram of angles}
\end{figure}

\subsection{Spectral Properties of $L$} \label{subsec: Spectral Properties of the Jacobian Matrix L}
%Recall that a matrix $A$ is an \emph{M-matrix} if the off-diagonal elements of $A$ are nonpositive and the nonzero eigenvalues of $A$ have positive real parts \cite{1994-Berman-Nonnegative-matrices}.
When $\varphi_{ij}$'s satisfy the above angle assumption, the following proposition from \cite{2020-fast-certificate} shows that $L$ is a singular M-matrix. Moreover, the zero eigenvalue of $L$ is simple, i.e. the algebraic and geometric multiplicities are one, which is important for preventing bifurcation from happening in the multi-$\mu$grid network. %The following proposition summarizes some key properties of $L$.  
\begin{proposition}\label{prop: quasi-M}
Let $(\delta^{*},\omega^{*})\in \Omega$ be an EP of the multi-$\mu$G system \eqref{eq: swing equations}. Assume the linking grid $\G$ is connected. The Jacobian matrix $L$ defined by \eqref{eq: Jacobian} at this EP has the following properties:
\begin{enumerate}[(i)]
	\item $L$ has a zero eigenvalue with an eigenvector $\mathbf{1}$, and the real part of each nonzero eigenvalue of $L$ is positive, i.e. $L$ is a singular M-matrix.
	
	\item The zero eigenvalue of $L$ is simple.
	
	%\item If the distribution network has a connected underlying undirected graph, then the zero eigenvalue of $L$ is simple.
	%Both algebraic and geometric multiplicities of the zero eigenvalue of $L$ are  equal to the number connected components of .
		
	%\item  $L =  \lambda_0 I - B$ for some nonnegative matrix $B$ (i.e., $b_{ij} \ge 0, \forall i, j$) and some $\lambda_0 \ge \rho$, where $\rho$ is the maximum eigenvalue of $B$.

	%\item All principal minors of $L$ are nonnegative.
	\end{enumerate}
%\textbf{I have to prove that an $L$ matrix which is weakly diagonally dominant is a quasi-M matrix }
% I need to prove that J11 is quasi-m matrix. Then I can use Theorem 3.1 of this paper.
%A Survey on M-Matrices
%Author(s): George Poole and Thomas Boullion
%Source: SIAM Review, Vol. 16, No. 4 (Oct., 1974), pp. 419-427
%Published by: Society for Industrial and Applied Mathematics
%Stable URL: https://www.jstor.org/stable/2028688
%Accessed: 29-10-2018 23:07 UTC
\end{proposition}
%
% \begin{proof}
% See Appendix \ref{sec: proof of quasi-M}.
% \end{proof}
%----------------------------------------
Properties (i) and (ii) of $L$ shown in the above proposition will be used in the next section to prove the stability of $J$ in the main result of the paper. 
\section{Stability of Multi-Microgrid Networks} \label{Sec: Stability and Hyperbolicity of the Equilibrium Points}
Now we are ready to answer the fundamental question posed in Section \ref{subsec: Multi-Machine Swing Equations}: under what conditions is an EP $(\delta^*,\omega^*)$ locally asymptotically stable?
\subsection{The Main Stability Theorem}\label{sec:main_stability_theorem}
\begin{theorem} \label{thm: stability properties}
    Let $(\delta^{*},\omega^{*})\in \Omega$ be an EP of the multi-$\mu$G system \eqref{eq: swing equations}. Let $B\in\mathbb{R}^{n\times n}$ denote the imaginary part of the admittance matrix. Suppose all microgrid interfaces have positive damping coefficient and inertia, and the linking grid $\G$ is connected. Then, the following statements hold:
    \begin{enumerate}[(a)]
        \item  The Jacobian $J$ at this EP has a zero eigenvalue with geometric multiplicity of one. \label{thrm-a}
        %with a one-dimensional invariant subspace $ \mathbf{span}(\mathbf{1})$.  
        \item All the nonzero real eigenvalues of $J$ are negative. \label{thrm-b}
        \item Let $Q_i$ be the net outgoing reactive power flow from microgrid PCC $i$. If 
        \begin{align} \label{eq: condition for stability}
            -Q_i - V_i^2 B_{ii} \le \frac{d_i^2}{2m_i}, &&    \forall i \in \mathcal{N} % \tag{\textbf{C}}
        \end{align}
        % $D_i^2-2M_iL_{ii}\ge0, \forall i \in \mathcal{N}$,
        then all the nonzero eigenvalues of $J$, both real and complex, are located in the left half plane, i.e., $\sigma(J)\subset\mathbb{C}_-\cup\{0\}$, and the EP is locally asymptotically stable. \label{thrm-c}
        
        \item If the transfer conductance of the lines is zero, then all the nonzero eigenvalues of $J$ are located in the left half plane, and the EP is locally asymptotically stable. \label{thrm-d}
         %and complex eigenvalues all with negative real parts.
    \end{enumerate}
\end{theorem}
%---------------
\begin{proof}
See Appendix \ref{app: proof of thm: stability properties}.
\end{proof}
\begin{remark}
Properties (\ref{thrm-a}) and (\ref{thrm-b}) hold independently of the sufficient conditions in (\ref{thrm-c}) and (\ref{thrm-d}). Property (\ref{thrm-d}) says if the network is lossless, then regardless of (\ref{thrm-c}), any EP is stable. If instead the network is lossy, then not every EP is stable and condition \eqref{eq: condition for stability} provides a new certificate to guarantee the small-signal stability of an EP. 
\end{remark}
\begin{remark}
    Notice that a salient feature of condition \eqref{eq: condition for stability} is that it only requires local information at each microgrid interface, hence, leads to a fully distributed control scheme to stabilize the multi-$\mu$G system. Detailed numerical simulation will be shown in Section \ref{Sec: Computational Experiments}.
\end{remark}
%----------------
%
%
\subsection{Intuition and Paradox Behind Condition \eqref{eq: condition for stability}}

Condition \eqref{eq: condition for stability} in Theorem \ref{thm: stability properties} provides a practical and efficient way to certify the stability of the EPs in general lossy multi-$\mu$G networks. It also introduces a distributed control rule for tuning the interface parameters of each microgrid without compromising the network stability.
%Moreover, it can be incorporated into the multi-$\mu$G scheduling problem as an additional constraint, thereby guaranteeing the transient stability of the solution. 
In this section, we will explore the intuition behind this theory as well as two interesting paradoxes that come with it.

%\begin{remark} \hfill
    \begin{itemize}

    \item \emph{\textbf{Note 1:}}
    The variable $Q_i$ in \eqref{eq: condition for stability} is the net reactive power that microgrid $i$ injects into the rest of the multi-$\mu$G network, that is, if microgrid $i$ is supplying reactive power, then $Q_i>0$. Otherwise, if it is consuming reactive power, then $Q_i<0$. Intuitively, when microgrid $i$ is a supplier of reactive power, the first term on the left-hand side of \eqref{eq: condition for stability} is negative, and this situation will help condition \eqref{eq: condition for stability} hold. 
    
    \item
    \emph{\textbf{Note 2:}} 
    Recall that $Y_{ii}\measuredangle \theta_{ii}=G_{ii}+jB_{ii} = \sum_{k=1}^n y_{ik}$, where $y_{ik}=g_{ik}+jb_{ik}$ is the admittance of line $(i,k)$, with $g_{ik}\ge0$ and $b_{ik}\le0$. Therefore, $B_{ii}\le0$, and the second term on the left-hand side of \eqref{eq: condition for stability} is always positive. Here, it is assumed that $y_{ii}$, i.e., the admittance-to-ground at PCC $i$ is negligible. Otherwise, we may have $B_{ii}>0$, and the second term on the left-hand side of \eqref{eq: condition for stability} could be negative.
    
    \item
    \emph{\textbf{Note 3:}}
    The first two notes clarify that the left-hand side of \eqref{eq: condition for stability} can be negative if microgrid $i$ is supplying reactive power; otherwise it is positive. Consequently, condition \eqref{eq: condition for stability} is not trivial.
    
    \item
    \emph{\textbf{Note 4:}}
    Condition \eqref{eq: condition for stability} enforces an upper bound which is proportional to the square of damping and inverse of inertia. This is consistent with the intuition that if we increase the damping, the stability margin of the system will increase. However, it is not intuitive (could be a paradox) that decreasing the virtual inertia of a microgrid interface will increase the stability margin.
    
    \item
    \emph{\textbf{Note 5:}} By adding more transmission lines to the system, $|B_{ii}|$ will increase, and this in turn could increase the left-hand side of \eqref{eq: condition for stability} and lead to the violation of this condition. This can be called the Braess's Paradox in power system stability.
    %
    %which reveals that adding a road to a particular congested traffic network would increase overall journey time.
    \end{itemize}
%\end{remark}
The following corollary will further illustrate the aforementioned Braess's Paradox.
\begin{corollary} \label{coro: paradox}
Under the assumptions of Theorem \ref{thm: stability properties}, if 
\begin{align} \label{eq: condition for stability without sin}
            \sum \limits_{k=1, k \ne i}^n { V_i V_k Y_{ik} } \le \frac{d_i^2}{2m_i},     && \forall i \in \mathcal{N}
            %\tag{\textbf{C-1}}
\end{align}
then the nonzero eigenvalues of $J$ are located in the left half plane and the EP is locally asymptotically stable.%i.e., $\sigma(J)\subset\mathbb{C}_-\cup\{0\}$.
\end{corollary}
%Recall that $Y_{ik}\measuredangle \theta_{ik} = - y_{ik}$, i.e., it is the negative of the admittance between the nodes $i$ and $j$.

This corollary directly follows from the proof of Theorem \ref{thm: stability properties} provided in Appendix \ref{app: proof of thm: stability properties}. Counterintuitively, according to \eqref{eq: condition for stability without sin}, adding more power lines can lead to violating the sufficient condition for stability. This Braess's Paradox in power systems has been also acknowledged in \cite{2015-robustness-paradox-network} and \cite{Vu-2019-reconfig-microgrids} in different contexts and using different approaches. 
Note that removing lines from a network could make the system more vulnerable to contingencies and eliminate the reliability benefits of having more transmission line capacity. This trade-off should be taken into account during the design and operation of power grids. 
%}

%------------------------------------------------------------------------

%\subsection{Structure-Preserving Higher-Order Microgrids}
\subsection{Stability Condition in Structure-Preserving Networks}
\label{subsec: structure-preserving microgrid}
%
%\section{Understanding The effect of Kron Reduction}
%
\subsubsection{Motivation}
%\textcolor{blue}{
The stability certificate \eqref{eq: condition for stability} in Section \ref{sec:main_stability_theorem} is derived for the linking grid
$\mathcal{G}=(\mathcal{N},\mathcal{E})$, where each microgrid is reduced to one node modeled as a grid-forming inverter using the swing equation. In this section, we consider the more general situation, where the internal active and passive elements of a microgrid are explicitly modeled. In particular, let $\mathcal{G}^d = (\mathcal{N}^d, \mathcal{E}^d)$ be the distribution network composed of all microgrids $\mathcal{G}_i=(\mathcal{N}_i, \mathcal{E}_i)$ for $i\in\mathcal{N}$ as subnetworks. The buses $\mathcal{N}_i$ of microgrid $i$ may include both active nodes (i.e., those connected to DGs and/or VSIs) and passive nodes (i.e., those connected to a constant admittance load). 

%\textcolor{blue}{
In order to study the stability property of $\mathcal{G}^d$, we first use Kron reduction to eliminate all passive nodes from each microgrid and study the resulting reduced network $\mathcal{G}^r$. The stability condition \eqref{eq: condition for stability} can be applied to $\mathcal{G}^r$. However, such a certificate is expressed in the system states and parameters of $\mathcal{G}^r$, not of the original network $\mathcal{G}^d$. This is not desirable, as it obscures the relations between the topology of the original network and the
stability properties of the EPs. Moreover, the parameters of the Kron-reduced network $\mathcal{G}^r$ may not be available to the individual microgrid controllers in real time. We want to find a stability certificate for the Kron-reduced network $\mathcal{G}^r$, expressed in the system states and network topology of the original network $\mathcal{G}^d$.
%}%

%\textcolor{blue}{
To tackle this challenge, we first identify in Section \ref{sec: Kron} a sufficient condition on the admittances of the original network $\mathcal{G}^d$, under which certain monotonic relationship between the admittances of $\mathcal{G}^d$ and  $\mathcal{G}^r$ can be obtained. Then in Section \ref{sec: stability_cert_originalnetwork}, we use this monotonicity property to derive a stability certificate expressed directly in states and parameters of $\mathcal{G}^d$.
%}

% establishing a monotonic relationship between parameters of the original network and those of the Kron-reduced network. We will then use this relationship to derive a stability condition as a function of the original network parameters.
%will transform the dynamical model \eqref{eq: swing equations} will transform the
%Assuming that the loads in each microgrid are modeled as constant admittances,
\subsubsection{The Kron-Reduced and Original Networks}\label{sec: Kron}
% Recall $\mathcal{G}_i = (\mathcal{N}_i, \mathcal{E}_i)$ denotes the internal power grid associated with microgrid $i\in\mathcal{N}$.
% Now, each microgrid $\mathcal{G}_i$ can be Kron reduced to network  $\mathcal{G}_i^r$ consisting of only active nodes, defined as follows:
%
%\textcolor{blue}{
\begin{definition} \label{def: Kron}
	Let $Y$ be the nodal admittance matrix of a microgrid $\mathcal{G}_i = (\mathcal{N}_i, \mathcal{E}_i)$, where the set of active and passive nodes are denoted by $\alpha,\beta \subseteq\mathcal{N}_i$, respectively. The Kron reduction of $\mathcal{G}_i$ that eliminates all nodes in $\beta$ has an admittance matrix given by $Y^{r} := Y[\alpha,\alpha] - Y[\alpha,\beta] Y[\beta,\beta]^{-1} Y[\beta,\alpha]$. This Kron-reduced network is denoted by $\mathcal{G}_i^r$.
\end{definition}
%}
%
%Recall that in small-signal stability studies, loads are considered constant admittances and reflected into the diagonal entries of the nodal admittance matrix $Y$. This inclusion of loads transforms the matrix $Y\in\mathbb{C}^{n\times n}$ from a Laplacian matrix to a loopy Laplacian matrix which has a nonzero row sum and induces a graph with self loops. This loopy Laplacian matrix can be better studied by introducing an additional \emph{ground node} and defining an augmented nodal admittance matrix where loads are connected to the ground node instead of being modeled as self-admittances. Such an augmented nodal admittance matrix has a zero row sum. Interestingly, this ground node augmentation and Kron reduction commute \cite{2012-dorfler-kron}, i.e., the order of these two transformations can change and they lead to the same result. Based on this property and without loss of generality, we assume the ground node augmentation is already performed and therefore the nodal admittance matrix has zero row sum. 

%\textcolor{blue}{
Assumption \ref{as: y_matrix sign} below is widely satisfied in distribution grids.
\begin{assumption} \label{as: y_matrix sign}
The nodal admittance matrix $Y=G+jB$ of a distribution grid  satisfies $G_{ik}\le0, B_{ik}\ge0,$ for all $i\ne k$, and $G_{ii}\ge0, B_{ii}\le0$ for self-admittances. 
\end{assumption}
%}
%
%
%\textcolor{blue}{
Assumption \ref{as: y_matrix bound} below is the sufficient condition used in Lemma \ref{lem: assumption closed Kron} to derive a monotonicity relation between the admittances of the original and Kron-reduced networks.
\begin{assumption} \label{as: y_matrix bound}
Let $Y=G+jB$ be the nodal admittance matrix of a distribution grid. There exist two fixed real numbers $\nu_{\min}$ and $\nu_{\max}$ that satisfy 
\begin{align}\label{eq:kron_nu}
    % & \sqrt{ (\nu_{\max}^2-1)/2}  \le \nu^i_{\min} \le \nu_{\max}.\\
    & 0\le \nu_{\min} \le \nu_{\max}\le \sqrt{1+ 2\nu^2_{\min}}
\end{align}
such that, for every line $(i,k)$, the conductance $G_{ik}$ and susceptance $B_{ik}$ are bounded as 
\begin{align}\label{eq:kron_bound}
    \nu_{\min} |G_{ik}| \le |B_{ik}| \le \nu_{\max} |G_{ik}|.
\end{align} 
\end{assumption}
\begin{remark}
%
%Assumption \ref{as: y_matrix bound} says that the ratio $|B_{jk}|/|G_{jk}|$ should be similar for all lines. In particular, 
By \eqref{eq:kron_bound}, if $G_{ik}=0$, then $B_{ik}=0$; otherwise, $\nu_{\min}\le |B_{ik}|/|G_{ik}|\le \nu_{\max}$, where the upper and lower bounds satisfy \eqref{eq:kron_nu}. As an example, if $\nu_{\min}=5$, then we can set $\nu_{\max}=\sqrt{1+2\cdot 5^2}=7.14$. Then, according to \eqref{eq:kron_bound}, all lines have $|B_{ik}|/|G_{ik}|$ ratio between $5$ and $7.14$, which is typical in distribution grids, especially in microgrids.
\end{remark}
%}
%
%Before proceeding to the main result of this section,
% The following lemma establishes a relationship between the Kron-reduced and original networks, and is proved in Appendix \ref{sec: proof of lem: assumption closed Kron}.
%
%
%
% \begin{lemma} \label{lem: assumption closed Kron}
%     Let $\mathcal{Y}$ and $\mathcal{Y}'$ be the classes of nodal admittance matrices that satisfy Assumptions \ref{as: y_matrix sign} and \ref{as: y_matrix bound}, respectively. Suppose $\mathcal{Y}'$ is invariant under Kron reduction, i.e., any admittance matrix $Y$ that satisfies Assumption \ref{as: y_matrix bound} will still satisfy this assumption after Kron reduction. 
%     Then $\mathcal{Y}\cap\mathcal{Y}'$ is also invariant under Kron reduction.
% \end{lemma}
%
%\textcolor{blue}{
\begin{lemma} \label{lem: assumption closed Kron}
Suppose the nodal admittance matrix $Y=G+jB$ of a distribution grid satisfies Assumptions \ref{as: y_matrix sign} and \ref{as: y_matrix bound} and the Kron-reduced matrix $Y^r=G^r+jB^r$ from eliminating a passive node $k_0\in\mathcal{N}_i$ satisfies Assumption \ref{as: y_matrix bound}. Then, $Y^r$ also satisfies Assumption \ref{as: y_matrix sign}. Furthermore, the monotonicity condition, $ B^{r}_{kk} \ge B_{kk}$, holds for all nodes $k\ne k_0$.
\end{lemma}
%}

%\textcolor{blue}{
See Appendix \ref{sec: proof of lem: assumption closed Kron} for the proof of this lemma.
%
%Now, the next proposition 
% \begin{proposition}
% Let $Y$ be a nodal admittance matrix that satisfies the assumptions of Lemma \ref{lem: assumption closed Kron}.
% we conclude that
% $
% 	B_{ii}^{r} \ge B_{ii},  \forall i \in \{1, \cdots, n-|\beta| \}
% $
% where $B_{ii}^{r}$ and $B_{ii}$ are the $i$th diagonal entries of the Kron-reduced and original admittance matrices, respectively.
% \end{proposition}
%
%---------------------------
%
%See Appendix \ref{sec: proof of lem: assumption closed Kron lossless}
%
% \begin{lemma} \label{lem: assumption closed Kron lossless}
%     Let $\mathcal{Y}$ be the class of lossless nodal admittance matrices that satisfy Assumption \ref{as: y_matrix sign}.
%     Then $\mathcal{Y}$ is closed under Kron reduction.
% \end{lemma}
%
%Lemma \ref{lem: assumption closed Kron} is %.
%and Lemma \ref{lem: assumption closed Kron lossless} can be proved similarly.
%
%-----------------

\subsubsection{Stability Condition as a Function of Original Network}\label{sec: stability_cert_originalnetwork}
% Recall that condition \eqref{eq: condition for stability} in Theorem \ref{thm: stability properties} is derived for the linking grid $\mathcal{G}=(\mathcal{N},\mathcal{E})$ where each node represents an entire microgrid.
% is represented by $\mathcal{G}_i = (\mathcal{N}_i, \mathcal{E}_i)$ with $|\mathcal{N}_i|=1$ and $\mathcal{E}_i=\emptyset$. 
%
%\textcolor{blue}{
Recall that the Kron-reduced network $\mathcal{G}^r$ is obtained by Kron-reducing all passive nodes in all the microgrids. So $\mathcal{G}^r$ only contains active nodes and its dynamical system is defined by model \eqref{eq: swing equations}, where each active node has a swing equation. The next theorem is the key result of this section that states a stability certificate for $\mathcal{G}^r$ but expressed in the states, network topology, and parameters of the original multi-$\mu$G, where microgrids are allowed to have an arbitrary internal structure with DGs, grid-forming inverters, and passive loads. 
%
%
%the Kron-reduced network. The next Theorem clarifies the transformation of the stability criterion in the original network.
%shows that under reasonable assumptions, condition (12) can be converted to a stability criterion in the original network.
\begin{theorem} \label{thrm: stability of original net}
%The following statements hold:
%\begin{enumerate}[(a)]
    % \item If 
    %     \begin{align} \label{eq: stability condition orig net Br}
    % 		-Q_k - V_k^2 B_{kk}^r \le \frac{d_k^2}{2m_k} && \forall k\in\alpha_i, i \in \mathcal{N},
    % 	\end{align}
    % 	then the EP is locally asymptotically stable. Here $B_{kk}^r$ is the $k$th diagonal entry of the imaginary part of the reduced admittance matrix.
    % \label{thrm: stability of original net - a}	
    %\item 
    Suppose Assumption \ref{as: y_matrix sign} holds for all the microgrids $\mathcal{G}_i=(\mathcal{N}_i, \mathcal{E}_i)$ for $i\in\mathcal{N}$ in the distribution grid $\mathcal{G}^d$ and Assumption \ref{as: y_matrix bound} holds for the reduced admittance matrix of $\mathcal{G}^d$ resulting from removing any passive node $k_0$ in $\mathcal{G}_i$ for any $i\in\mathcal{N}$. Then an EP of the Kron-reduced grid $\mathcal{G}^r$ is locally asymptotically stable, if the following condition holds
   \begin{align} \label{eq: stability condition orig net B}
    		-Q_k - V_k^2 B_{kk} \le \frac{d_k^2}{2m_k}, \quad \forall k\in\alpha_i, i \in \mathcal{N},
    	\end{align}
    where $\alpha_i\subseteq\mathcal{N}_i$ is the set of active nodes in microgrid $\mathcal{G}_i$ and all quantities $Q_k, V_k, B_{kk}, d_k, m_k$ correspond to the original network $\mathcal{G}^d$.
    %\label{thrm: stability of original net - b}
%\end{enumerate}
%
% 	Under Assumption \ref{as: y_matrix}, if
% 	\begin{align} \label{eq: stability condition orig net}
% 		-Q_i - V_i^2 B_{ii} \le \frac{d_i^2}{2m_i} && \forall i \in \mathcal{N}\setminus \beta
% 	\end{align}
% holds, then condition (12) in Theorem 1 holds for the Kron-reduced network. Note that	in \eqref{eq: stability condition orig net}, the quantities $Q_i, V_i, B_{ii}, d_i, m_i$ correspond to the original network.
\end{theorem}
The proof of this theorem is given in
Appendix \ref{sec: proof of thrm: stability of original net}.
\section{Computational Experiments} \label{Sec: Computational Experiments}
In this section, we test various aspects of Theorems \ref{thm: stability properties} and \ref{thrm: stability of original net}, and show how they can be used not only as a fast stability certificate, but also as a quantitative measure of the degree of stability. Furthermore, we demonstrate that condition \eqref{eq: condition for stability} offers a distributed control rule to retain and ensure the stability of interconnected microgrids in an emergency situation. Let us define
\begin{align*}
    \mathcal{S}_i := -Q_i - V_i^2 B_{ii} - \frac{d_i^2}{2m_i},
\end{align*}
and recall that according to condition \eqref{eq: condition for stability} in Theorem \ref{thm: stability properties}, if $\mathcal{S}_i\le0, \forall i\in\mathcal{N}$, then the EP of the multi-$\mu$G system is guaranteed to be asymptotically stable.
\subsection{Control Schemes and Braess's Paradox}
Consider the four-microgrid system shown in Fig. \ref{fig: schematic diagram 4 mic}a and its load-flow and dynamical data tabulated in Case (a1) of Table \ref{tab: Dynamic Parameters 4-mic}. The system is normally operating in this case, but $\mathcal{S}_i>0, \forall i\in\mathcal{N}$ and Theorem \ref{thm: stability properties} does not certify the stability of the system. Such a positive $\mathcal{S}_i$ for all microgrids indicates that the multi-$\mu$G system, albeit operating normally, is close to its stability margins. We will show how a credible contingency could push such an uncertified system into instability.

\noindent\textbf{Internal Outage Leads to Instability:} Subsequent to a generation outage inside microgrid $\mu G_4$, the active power $P_{s_4}$ changes from $-4.06$ to $-7.06$, i.e., this microgrid starts to get $3$ p.u. more active power from the linking grid to compensate for its internal outage. In response, microgrid $\mu G_3$ aids $\mu G_4$ by using its internal generation capacity and changing its active power $P_{s_3}$ from $-2.25$ to $0.25$. See Case (a2) in Table \ref{tab: Dynamic Parameters 4-mic}.
Such a smart, resilient, and self-healing multi-$\mu$G system seems very appealing and is indeed one of the main purposes of building these interconnected systems. However, as it was hinted by positive values of $\mathcal{S}_i$ (i.e., violation of condition \eqref{eq: condition for stability}), this new EP of the multi-$\mu$G system is unstable. The instability of this EP can be verified through eigenvalue analysis and time domain simulation, as depicted in Fig. \ref{fig: unstable four microgrid}. Now, Theorem \ref{thm: stability properties} offers two remedial approaches to ensure system stability.

%The  load flow data as well as the dynamic parameters of the interface controllers are provided in Table \ref{tab: Dynamic Parameters 4-mic}. 
%
%In this state of the multi-$\mu$G system, the EP $(\delta^{*},\omega^{*})$ is located in the set $\Omega$, but condition \eqref{eq: condition for stability} is not satisfied and the EP is unstable as shown in Fig. \ref{fig: unstable 4-mic numerical simulation}. 

%Now, Theorem \ref{thm: stability properties} provides a fast, decentralized, and insightful approach to make the system stable. 

%a fast, decentralized, and insightful approach to make the system stable.

\noindent\textbf{A Distributed Control Scheme:} 
The first approach is based on a distributed control rule instructing how to change the interface controller parameters $d_i$ or $m_i$ in order to improve the multi-$\mu$G stability
(recall the characterization of $m_i$ and $d_i$ for microgrids described in Section \ref{subsec: Multi-Machine Swing Equations}).
Based on local measurements of reactive power $Q_i$ and voltage $V_i$, each microgrid can increase its damping $d_i$ and/or decrease its virtual inertia to ensure that condition \eqref{eq: condition for stability} is satisfied.
%The idea behind this approach is to increase the upper bound $\frac{d_i^2}{2m_i}$ in \eqref{eq: condition for stability} until $\mathcal{S}_i\le0$ is satisfied. 
The key features of the distributed control scheme include 1) by increasing $d_i^2/m_i$ the system can always be stabilized according to condition \eqref{eq: condition for stability}; 2) the operating point of the system is not changed; 3) no information exchange from the neighboring microgrids is required.  Implementing this approach, we reach to Case (a3) in Table \ref{tab: Dynamic Parameters 4-mic}. The stability of the same EP as in Case (a2) is certified.

%economixcally optimal
%the EP does not change

\noindent\textbf{A Coordinated Control Scheme:} 
The second approach offers coordination of a more general set of corrective actions 
%Accordingly, each microgrid $i\in\mathcal{N}$ should have $\mathcal{S}_i\le0$, and this can be achieved through one or a combination of the available options 
including change of interface controller parameters $d_i$ or $m_i$, change of reactive power $Q_i$ or voltage magnitude $V_i$, and change of network topology. Condition \eqref{eq: condition for stability} instructs which actions will improve the stability of the EP. The EP of the system may be moved in the coordinated control scheme to achieve corrective actions with smaller magnitude.
%The advantage of this approach is that through such a wide range of control actions the EP of the system could be moved to an (economically) optimal position without jeopardizing the system stability. On the downside, the approach may not be implementable in a fully distributed manner.  
%
To illustrate, we choose a combination of all available options to find a stable EP. Let us reconfigure the network by switching two lines off (see Fig. \ref{fig: schematic diagram 4 mic}(b)) and also modify the dynamic parameters to reach Case (b) in Table \ref{tab: Dynamic Parameters 4-mic}. The new EP satisfies condition \eqref{eq: condition for stability} and therefore is stable. Note that by removing distribution lines from case (a), the value of $|B_{ii}|, \forall i \in \{1,2,4\}$ will decrease. Moreover, increasing damping and decreasing inertia will increase the right-hand side of \eqref{eq: condition for stability}. Consistent with Braess's Paradox, switching off two lines indeed improves system stability. %, consistent with the guidelines suggested by condition \eqref{eq: condition for stability}.
%
%lead to the new stable EP, certified by condition \eqref{eq: condition for stability}. 
%Note also that the proposed stability certificate can be tested at each microgrid individually, therefore it is fully decentralized and does not require communication between microgrids, thereby making it even more reliable and resilient for real-time applications.
%
\begin{figure}[t]
 \includegraphics*[width=2.5in, keepaspectratio=true]{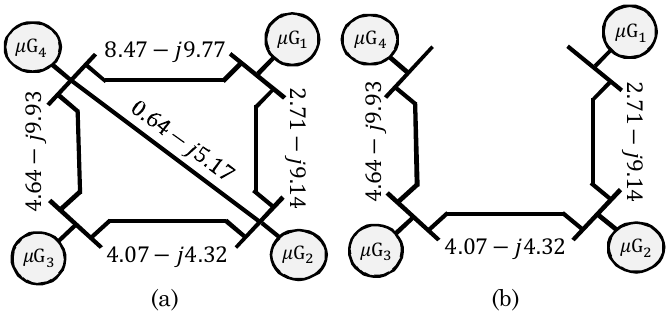}
 \centering
  \caption{Schematic diagram of four coupled microgrids.}
  \label{fig: schematic diagram 4 mic}
\end{figure}
%
% \begin{table} [t]
% \centering
% \caption{Dynamic parameters and converged load flow data of the four-microgrid case (a).}
% \label{tab: Dynamic Parameters 4-mic}
% \begin{tabular}{|c||c|c|c|c|c|}
% \hline 
% $\mu$G   & $m_i$ [sec.]    & $d_i$   & $P_{s_i}$ [p.u.]    & $V_i$ [p.u.]      & $\delta_i^*$ [rad]   \\ \hline \hline
% $i=1$   & $5.76$     & $1.03$          & $13.13$        & $0.963$    &  $0.477$    \\ \hline 
% $i=2$   & $9.2$      & $1.61$          & $0.39$       & $0.953$     & $0.072$ \\ \hline
% $i=3$   & $9.32$     & $1.86$          & $0.25$      & $0.997$      &   $-0.251$   \\ \hline
% $i=4$   & $4.92$     & $1.5$           & $-7.06$      & $1.02$      &   $-0.375$   \\ \hline 
% \end{tabular}
% \end{table}
%--------
% \begin{table} [t]
% \centering
% \caption{Dynamic parameters and converged load flow data of the four-microgrid case (b).}
% \label{tab: Dynamic Parameters 4-mic stable}
% \begin{tabular}{|c||c|c|c|c|c|}
% \hline
% $\mu$G   & $m_i$ [sec.]    & $d_i$   & $P_{s_i}$ [p.u.]    & $V_i$ [p.u.]      & $\delta_i^*$ [rad]   \\ \hline \hline
% $i=1$   & $0.8$     & $4.03$         & $5.72$      & $1.05$    &  $0.8$    \\ \hline 
% $i=2$   & $0.56$      & $3.9$        & $0.4$       & $1.05$     & $0.248$ \\ \hline
% $i=3$   & $0.7$     & $3.78$         & $0.25$      & $1.05$      &   $-0.623$   \\ \hline
% $i=4$   & $0.68$     & $3.49$        & $-2.11$     & $0.95$      &   $-0.8$   \\ \hline
% \end{tabular}
% \end{table}
%--------
\newcommand{\STAB}[1]{\begin{tabular}{@{}c@{}}#1\end{tabular}}
\begin{table}[t]
    \centering
    \caption{Dynamic parameters and converged load flow data of the four-microgrid system.}
    \label{tab: Dynamic Parameters 4-mic}
    \begin{tabular}{|c|c||c|c|c|c|c|c|}
    \hline 
     & \multicolumn{1}{c||}{$i$} & \multicolumn{1}{c|}{$m_i$} & \multicolumn{1}{c|}{$d_i$} & \multicolumn{1}{c|}{ $P_{s_i}$ [p.u.] } & \multicolumn{1}{c|}{$V_i$ [p.u.]}  &  \multicolumn{1}{c|}{$\delta_i^*$ [rad]} & \multicolumn{1}{c|}{$\mathcal{S}_i$}  \\
     \hline \hline
     %--------------
    \multirow{4}{*}{\STAB{\rotatebox[origin=c]{90}{Case (a1)}}}
     & $1$   & $5.76$     & $1.03$          & $13.13$      & $0.95$      &   $0.75$    & $21.18$  \\ \cline{2-8} 
     & $2$   & $9.20$      & $1.61$          & $0.39$       & $0.95$      &   $0.28$    &  $16.85$ \\ \cline{2-8} 
     & $3$   & $9.32$     & $1.86$          & $-2.25$      & $1.05$      &   $-0.18$   &  $12.12$ \\ \cline{2-8} 
     & $4$   & $4.92$     & $1.50$           & $-4.06$      & $1.05$      &   $-0.07$   &  $16.12$ \\ \hline  \hline
     %--------------
     \multirow{4}{*}{\STAB{\rotatebox[origin=c]{90}{Case (a2)}}}
     & $1$   & $5.76$     & $1.03$          & $13.13$        & $0.96$    &  $0.47$   & $21.17$ \\ \cline{2-8} 
     & $2$   & $9.20$      & $1.61$          & $0.39$       & $0.95$     & $0.07$    & $16.50$ \\ \cline{2-8} 
     & $3$   & $9.32$     & $1.86$          & $0.25$      & $0.99$      &   $-0.25$  & $13.08$ \\ \cline{2-8} 
     & $4$   & $4.92$     & $1.50$           & $-7.06$      & $1.02$      &   $-0.37$& $13.53$  \\ \hline  \hline
     %-----------
     \multirow{4}{*}{\STAB{\rotatebox[origin=c]{90}{Case (a3)}}}
     & $1$   & $0.50$     & $4.62$          & $13.13$        & $0.96$    &  $0.47$   & $-0.074$ \\ \cline{2-8} 
     & $2$   & $0.56$     & $4.32$          & $0.39$       & $0.95$     & $0.07$    &  $-0.035$ \\ \cline{2-8} 
     & $3$   & $0.66$     & $4.19$          & $0.25$      & $0.99$      &   $-0.25$  & $-0.037$ \\ \cline{2-8} 
     & $4$   & $0.56$     & $3.92$           & $-7.06$      & $1.02$      &   $-0.37$& $-0.001$  \\ \hline  \hline
     %-----------
     \multirow{4}{*}{\STAB{\rotatebox[origin=c]{90}{Case (b)}}}
    & $1$   & $0.80$     & $4.03$         & $5.72$      & $1.05$    &  $0.8$      &  $-0.0036$   \\ \cline{2-8}
    & $2$   & $0.56$      & $3.90$        & $0.40$       & $1.05$     & $0.24$    &  $-0.0668$   \\ \cline{2-8}
    & $3$   & $0.70$     & $3.78$         & $0.25$      & $1.05$      &   $-0.62$ &  $-0.0057$    \\ \cline{2-8}
    & $4$   & $0.68$     & $3.49$        & $-2.11$     & $0.95$      &   $-0.8$   &  $-0.0205$    \\ \hline 
  \end{tabular}
\end{table} 
%------
%---------------------------------
%
\begin{table} [t]
\centering
\caption{Parameters to generate synthetic networks. $U([\ell_1,\ell_2])$ is uniform distribution  on interval $[\ell_1,\ell_2]$.}
\label{tab: Parameters to generate synthetic networks}
\begin{tabular}{|c||c|}
\hline
Admittances  & $b={U}([-1,0])$ [p.u.], $g=|b|\times{U}([0,0.5])$ [p.u.]  \\ \hline 
Voltages     & $V={U}([0.95,1.05])$ [p.u.], $\delta={U}([-0.5,0.5])$ [rad]        \\ \hline 
Interface settings     & $d={U}([1.5,3])$, $m={U}([0.4,2])$ [sec.]    \\ \hline
\end{tabular}
\end{table}
%----------------------------------------------------------
%
% \vspace{7mm}
% \begin{figure}[t]
% \centering
% \begin{subfigure}[b]{0.5\textwidth}
%   \includegraphics*[width=3.4in, keepaspectratio=true]{figures//numer_4_mic.pdf}
%   \caption{}
%   \label{fig: unstable 4-mic numerical simulation} 
% \end{subfigure}
% \begin{subfigure}[b]{0.5\textwidth}
%   \includegraphics*[width=3.4in, keepaspectratio=true]{figures//numer_4_mic_stabilized.pdf}
%   \caption{}
%   \label{fig: stabilized 4-mic numerical simulation}
% \end{subfigure}
% \caption{Voltage angles $\delta$ in radians, and the deviation of the angular frequency $\omega$ in radians per seconds in the four-microgrid system. Both cases have the same initial condition.}
% \end{figure}
% \vspace{4mm}
%---------
\begin{figure}
\centering
\begin{subfigure}{0.26\textwidth}
  \centering
  \includegraphics[width=\linewidth]{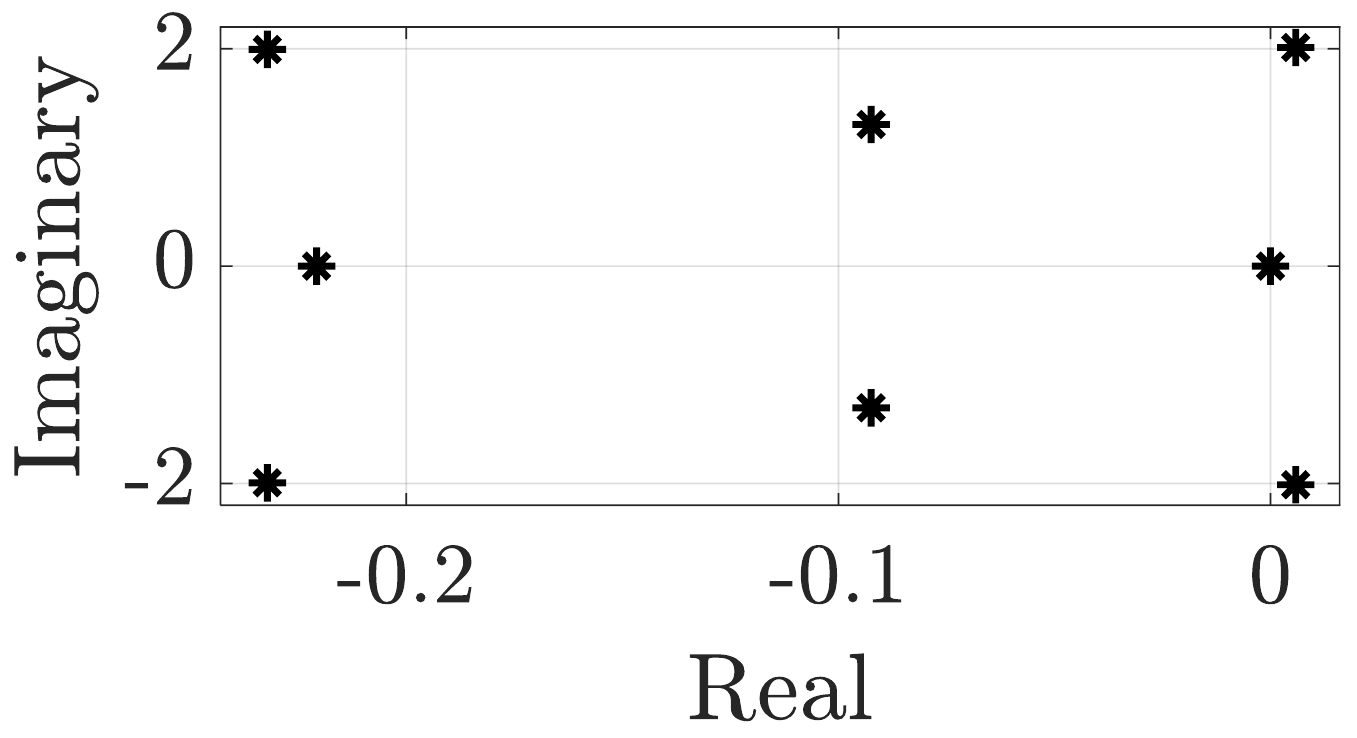}
  \caption{Spectrum of matrix $J$.}
  \label{fig: sfig1 unstable eigenvalues}
\end{subfigure}%
\begin{subfigure}{0.24\textwidth}
  \centering
  \includegraphics[width=\linewidth]{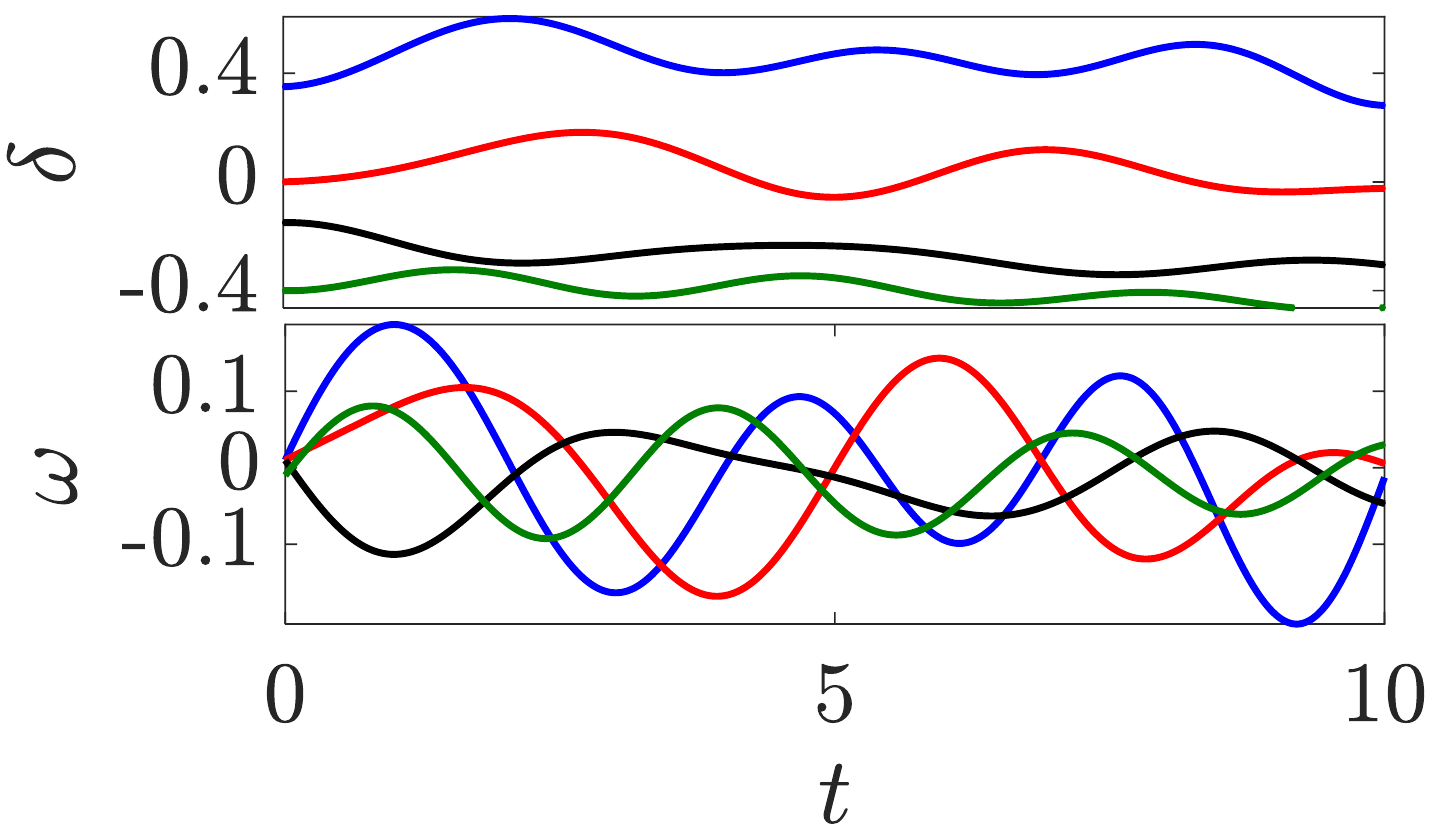}
  \caption{Orbits of the system.}
  \label{fig: sfig2 unstable time sim}
\end{subfigure}
\caption{Verifying the instability of the EP in Case (a2) of the four-microgrid system. (a) There exist two eigenvalues with positive real part. (b) Starting from a neighborhood of the EP, the orbits of the system diverge to infinity.}
\label{fig: unstable four microgrid}
\end{figure}
%
%----------------------------

\subsection{Stability Measure and Location of Eigenvalues}
%\textcolor{blue}{
As mentioned above, condition \eqref{eq: condition for stability} can be used not only as a fast stability certificate, but also as a quantitative measure of the degree of stability. To further illustrate this, consider the IEEE 33-bus network during islanded operation, consisting of $4$ DGs and $2$ storage units interfaced via VSIs \cite{2017-gholami-proactive}.
%$6$ controllable DERs  interfaced via VSIs.
The load, line, and DG data can be found in \cite{matpower, 1989-Baran-network, 2017-gholami-proactive}, and Table \ref{tab: ieee33 param}. Here, we first compute the Kron-reduced system to obtain a network of interconnected DGs. Note that Theorem \ref{thrm: stability of original net} is applicable to this reduced network because by Lemma \ref{lemma: VSC model reparametrization} the dynamical model of interconnected droop-controlled VSIs can be reparametrized as the swing equation model \eqref{eq: swing equations}. Observe that according to Fig. \ref{fig: histogram of angles}, the assumption $\varphi_{ik} \in (0,\pi)$ holds in this system. We assume the network is operating at $80$\% of the nominal load, and the interface parameters $k_i$, $\tau_i$, $m_i, d_i$, and setpoints are designed following Theorem \ref{thrm: stability of original net} (see Table \ref{tab: ieee33 param}). The simulations are carried out in \textsc{Matlab}.
%}

%\textcolor{blue}{
%
Fig. \ref{fig: 33-bus stability index and eigens} shows the spectrum of matrix $J$ along with the value of $\mathcal{S}_i, \forall i\in\{1,...,6\}$ under three different operating points referred to as Cases $1$ to $3$.  In Case $1$, $\mathcal{S}_i>0$ for $i=4$ and $i=6$. Moreover, in Case $2$, $\mathcal{S}_i>0$ for $i=6$. Case $3$ is the only case where $\mathcal{S}_i\le0, \forall i\in\{1,...,6\}$, and condition \eqref{eq: stability condition orig net B} guarantees that the system is asymptotically stable in this case. 
According to this figure, in all three cases the non-zero eigenvalues of $J$ are located in the left half plane and the system is asymptotically stable. However, from Case $1$ to Case $3$, as we move towards satisfying $\mathcal{S}_i\le0, \forall i\in\{1,...,6\}$, the magnitude of the imaginary parts of the eigenvalues in $\sigma(J)$ is reduced, and their real parts are
mainly moved towards $-\infty$, thereby making the system less oscillatory. 
Indeed, a smaller value of $\mathcal{S}_i$ (say when $\mathcal{S}_i>0$) means the violation of constraint $\mathcal{S}_i \le 0$ is smaller, and it is easier to enforce condition \eqref{eq: stability condition orig net B}, and therefore to make sure we have reached stability. 
Evidently, the value of $\mathcal{S}_i$ can be seen as a stability measure, i.e., it roughly indicates how stable the system is. This application of condition \eqref{eq: stability condition orig net B} was also shown in the four-microgrid test case in the previous section.
%}
%
%\textcolor{blue}{
Finally, Fig. \ref{fig: 33-bus freq} depicts the frequency trajectories of the system in Case $3$, where condition \eqref{eq: stability condition orig net B} holds. As can be seen, after a transient, all frequency deviations converge to zero, and the EP, which was certified by Theorem \ref{thrm: stability of original net}, is asymptotically stable. The initial condition in this simulation is chosen arbitrary within a reasonable range.
%}
%-------------------

\begin{table}[]
\centering
\caption{Parameters of the IEEE 33-bus system.}
\label{tab: ieee33 param}
\begin{tabular}{|c||c|c|c|c|c|c|}
\hline 
$i$         & $1$   & $2$   & $3$   & $4$   & $5$   & $6$  \\ \hline \hline
Bus index   & $8$ & $13$& $16$& $19$& $25$& $26$  \\ \hline
DER type    & DG  & DG  & DG  & VSI  & DG  & VSI \\ \hline
$d_i$       & $1.7$  & $1.7$  & $2$  & $1$  & $2$  & $1.2$ \\ \hline
$m_i$       & $0.5$  & $0.5$  & $0.6$  & $0.7$  & $0.6$  & $0.7$ \\ \hline \hline
Base values & \multicolumn{6}{c|}{$P_{\text{base}}=100~\text{MW}$, $V_{\text{base}}=12.66~\text{kV}$ } \\ \hline
\end{tabular}
\end{table}

%-------------------
%
\begin{figure}[]
\centering
\begin{subfigure}[b]{0.5\textwidth}
\begin{tikzpicture}
\begin{axis}[%
width=3.4in,
height=1.25in,
%at={(0.83in,0.587in)},
%scale only axis,
xmin=-2.5,
xmax=0,
xlabel style={font=\small},
xlabel={Real},
ymin=-1,
ymax=1,
ylabel style={font=\small},
tick label style={font=\small},
ylabel={Imaginary},
%axis background/.style={fill=white},
xmajorgrids,
ymajorgrids,
legend style={at={(0.02,0.63)}, anchor=south west, legend columns=3, legend cell align=left, align=left, draw=white!15!black, font=\footnotesize}
]
\addplot [color=blue, line width=1.0pt, only marks, mark size=2.5pt, mark=o, mark options={solid, blue}]
  table[row sep=crcr]{%
-0.586068523172822	0.910338836769889\\
-0.586068523172822	-0.910338836769889\\
-1.64499577513385	0\\
-1.52567198991446	0\\
-1.21021898531869	0\\
-0.816445660549694	0.500697323714555\\
-0.816445660549694	-0.500697323714555\\
-0.414160216094267	0.49974683185105\\
-0.414160216094267	-0.49974683185105\\
-0.23421121236361	0\\
2.12135495973183e-16	0\\
-0.0563151423977254	0\\
};
\addlegendentry{Case 1}

\addplot [color=red, line width=1.0pt, only marks, mark size=2.5pt, mark=triangle, mark options={solid, red}]
  table[row sep=crcr]{%
-2.02142134625104	0\\
-1.9084095668771	0\\
-1.62450682577512	0\\
-0.674536437129843	0.88922569823875\\
-0.674536437129843	-0.88922569823875\\
-1.00257346885259	0.125650978905842\\
-1.00257346885259	-0.125650978905842\\
-0.495095408083573	0.46690483496436\\
-0.495095408083573	-0.46690483496436\\
-0.200854169892951	0\\
-1.54627203174131e-16	0\\
-0.0506620133363201	0\\
};
\addlegendentry{Case 2}

\addplot [color=brown, line width=1.0pt, only marks, mark size=2.5pt, mark=asterisk, mark options={solid, brown}]
  table[row sep=crcr]{%
-2.49098366081536	0\\
-2.38193251794339	0\\
-2.13132663272671	0\\
-1.85412912937072	0\\
-0.781499332535502	0.859520240788613\\
-0.781499332535502	-0.859520240788613\\
-0.596297661723236	0.402855606168084\\
-0.596297661723236	-0.402855606168084\\
-0.619578471963335	0\\
-0.177504690177251	0\\
3.98939538645847e-16	0\\
-0.0460937656286247	0\\
};
\addlegendentry{Case 3}
\end{axis}

\begin{axis}[%
width=3.4in,
height=1.5in,
%at={(0in,0in)},
%scale only axis,
xmin=0,
xmax=1,
ymin=0,
ymax=1,
axis line style={draw=none},
ticks=none,
axis x line*=bottom,
axis y line*=left
]
\end{axis}
\end{tikzpicture}%
%\caption{}
%\label{fig: 33-bus eigenvalues}
\end{subfigure}
%
%\vspace{-5pt}
%-------------
\begin{subfigure}[b]{0.5\textwidth}
   \begin{tikzpicture}
\begin{axis}[
	x tick label style={/pgf/number format/1000 sep=},
	xlabel style={font=\small},
	ylabel style={font=\small},
    tick label style={font=\small},
	ylabel=$\mathcal{S}_i$,
	xlabel=$i$,
	legend style={at={(0.5,1.25)}, anchor=north,legend columns=-1, font=\footnotesize },
	ybar,
	bar width=5pt,
	scaled ticks=true,
	width=3.4in,
	height=1.25in
]

\addplot coordinates {(1,-0.519571111197457)	(2,-0.809441061604676)	(3,-1.21503045346939)	(4,0.0973912144208389)	(5,-1.38294950694226) (6,0.411364951490601)};
\addplot coordinates {(1,-1.01729333341968)	(2,-1.30716328382690)	(3,-1.78910452754346)	(4,-0.0256246585950342)	(5,-1.95702358101633)	(6,0.234222094347744)};
\addplot coordinates {(1,-1.67557111119746)	(2,-1.96544106160468)	(3,-2.54836378680272)	(4,-0.188323071293447)	(5,-2.71628284027559)	(6,-6.36199379702163e-05)};
\legend{Case 1, Case2 , Case 3}

\addplot[cyan,sharp plot,update limits=false] 
	coordinates {(0,0) (7,0)} ;

\end{axis}
\end{tikzpicture}
   %\caption{}
   %\label{fig: 33-bus stability index}
\end{subfigure}
\caption{ Illustration of stability certificate on the IEEE 33-bus system. 
(a) Spectrum of matrix $J$. (b) Value of stability index $\mathcal{S}_i$ in different buses.}
\label{fig: 33-bus stability index and eigens}
\end{figure}
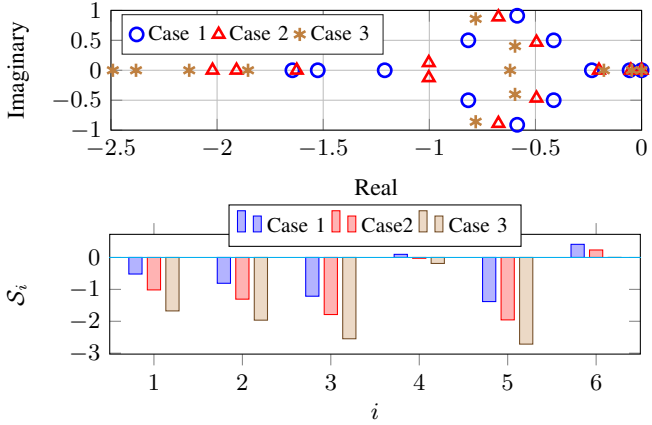

%------------------------------------------------------------------------

\begin{figure}
\centering
  \centering
  \includegraphics[width=3.5in, keepaspectratio=true]{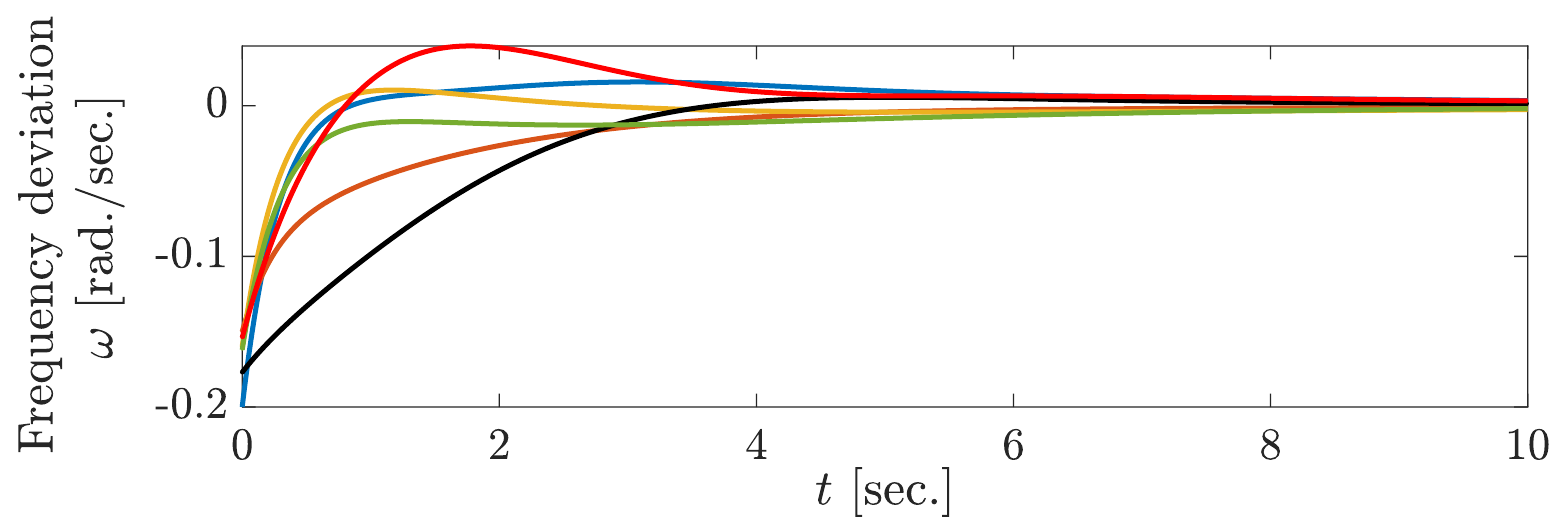}
  \caption{Trajectories of the frequency deviation $\omega_i$ for $6$ DERs in the IEEE $33$-bus system.}
  \label{fig: 33-bus freq}
\end{figure}
\vspace{-2mm}

% \begin{table}[]
% \centering
% \caption{Illustration of the proposed stability certificate in Theorem \ref{thm: stability properties}.}
% \label{tab: diffeent ieee cases}
% \begin{tabular}{|c|c|c|c|}
% \hline
% \textbf{Test~case}    &  $\min_i\{\mathcal{S}_i\}$  & $\max_i\{\mathcal{S}_i\}$ & \textbf{$\Re{(\lambda_2)}$} \\ \hline
% nesta-case7-utl       & $82 $     & $-0.7  $       & $3.18 $        \\ \hline
% nesta-case22-utl      & $0.43 $   & $-5.0  $       & $2.17 $        \\ \hline
% nesta-case30-utl      & $0.36 $   & $-12.2 $       & $0.75 $        \\ \hline
% ieee-case33           & $0.37$    & $-7.73  $       & $4.95 $        \\ \hline
% \end{tabular}
% \end{table}

%--------------------------------
\subsection{Larger-Scale Systems}
%\noindent\textbf{Larger-Scale Systems:} 
Next, we test the proposed stability certificate on a set of large-scale synthetic networks. Figs. \ref{fig: 50-microgrid network} and \ref{fig: 100-microgrid network} show two examples of such multi-$\mu$G networks consisting of $50$ and $100$ microgrids, respectively. The network graphs are randomly generated, the sparsity patterns of their adjacency matrix are depicted in Figs. \ref{fig: adj 50-microgrid network} and \ref{fig: adj 100-microgrid network}, and the corresponding static and dynamic parameters are given Table \ref{tab: Parameters to generate synthetic networks}. Note also that the diameter (i.e., the longest graph geodesic) of the graphs \ref{fig: 50-microgrid network} and \ref{fig: 100-microgrid network} are $6$ and $8$, respectively.
Adopting the aforementioned distributed control rule, each microgrid adjusts its controller parameters $d_i$ and $m_i$ (within the permissible range) to meet condition \eqref{eq: condition for stability}. Obeying this rule at an EP guarantees that all nonzero eigenvalues of the Jacobian matrix $J$ have negative real part, and consequently, the EP is locally asymptotically stable (see Figs. \ref{fig: eigens 50-microgrid network} and \ref{fig: eigens 100-microgrid network}).  
%
%part as depicted in Figs. \ref{fig: eigens 50-microgrid network} and \ref{fig: eigens 100-microgrid network}
%, and consequently, . %(see Corollary \ref{coro: stability of EP set}).
%
\begin{figure*}[]
\centering
\begin{subfigure}[b]{0.144\textwidth}
   \includegraphics*[width=\textwidth]{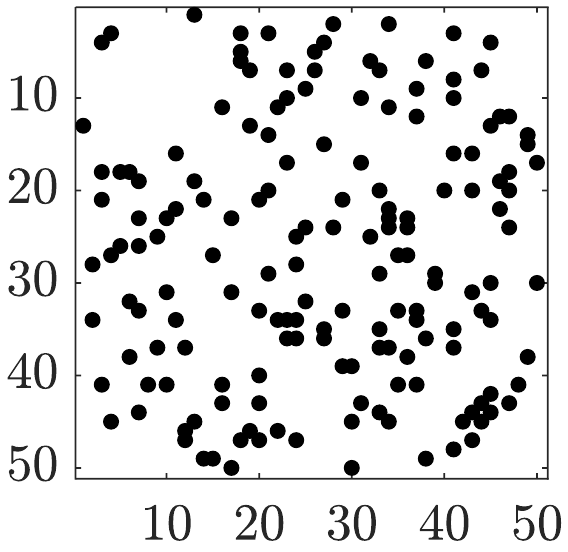}
   \caption{}
   \label{fig: adj 50-microgrid network}
\end{subfigure}
\begin{subfigure}[b]{0.21\textwidth}
   \includegraphics*[width=\textwidth]{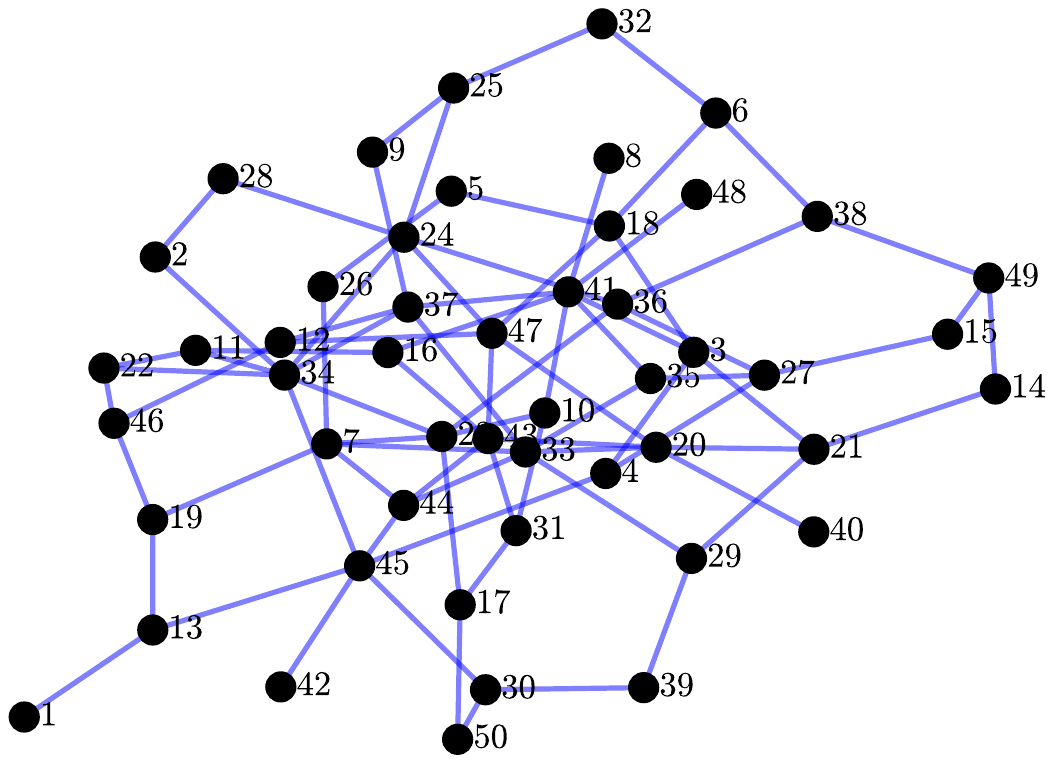}
   \caption{}
   \label{fig: 50-microgrid network} 
\end{subfigure}
\begin{subfigure}[b]{0.34\textwidth}
   \includegraphics*[width=\textwidth]{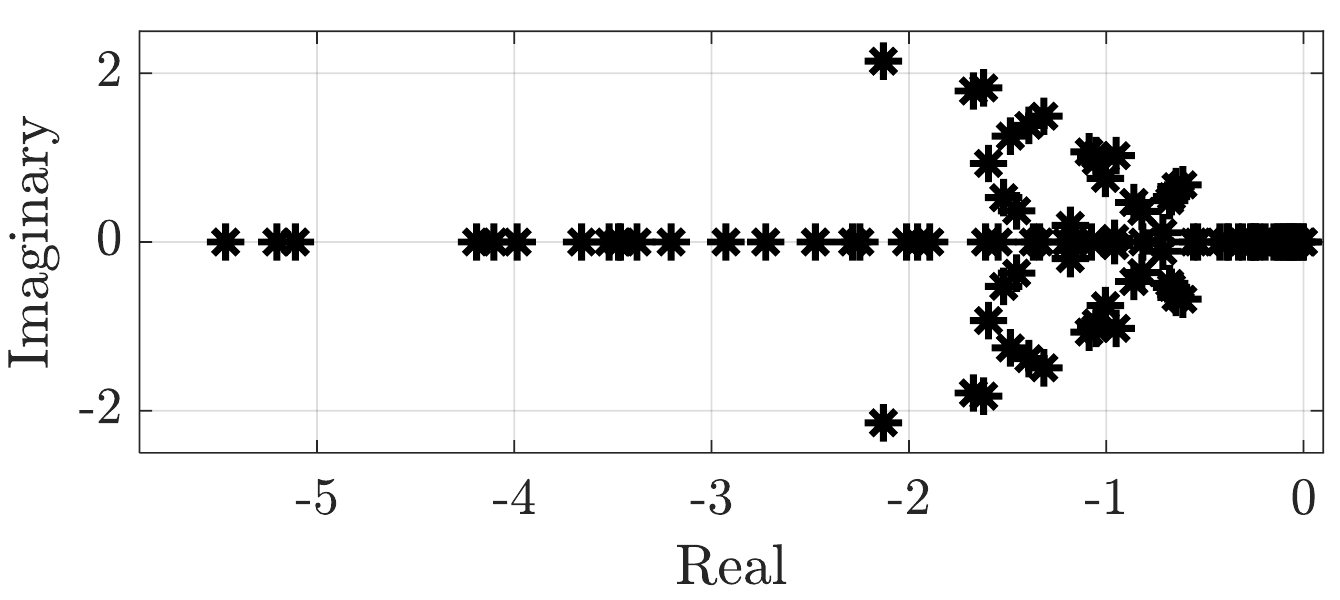}
   \caption{}
   \label{fig: eigens 50-microgrid network}
\end{subfigure}
%
%

%---------------
\begin{subfigure}[b]{0.151\textwidth}
   \includegraphics*[width=\textwidth]{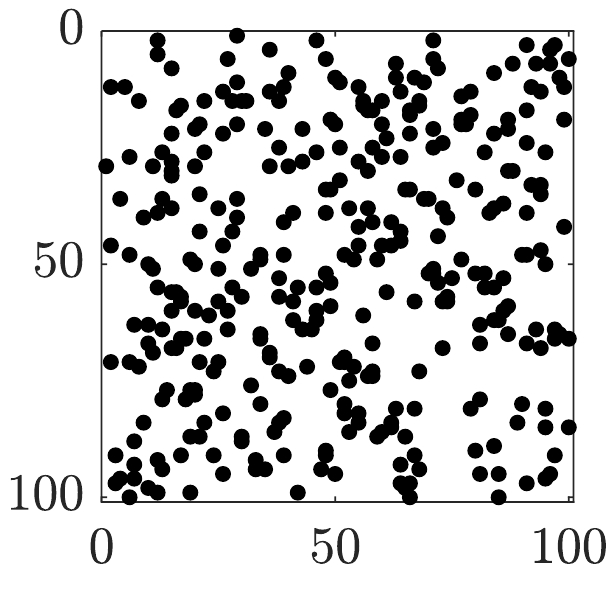}
   \caption{}
   \label{fig: adj 100-microgrid network}
\end{subfigure}
\begin{subfigure}[b]{0.21\textwidth}
   \includegraphics*[width=\textwidth]{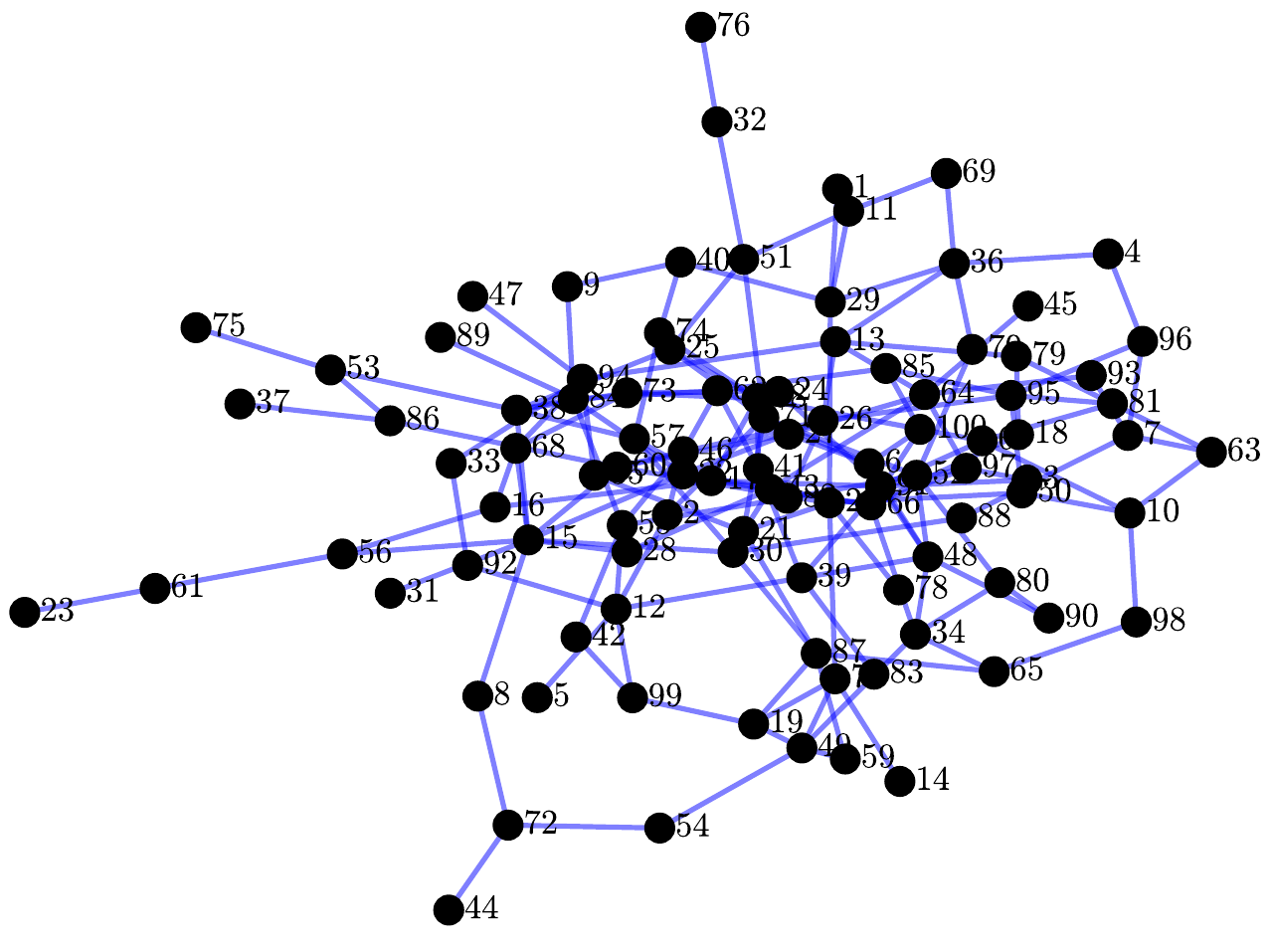}
   \caption{}
   \label{fig: 100-microgrid network}
\end{subfigure}
\begin{subfigure}[b]{0.34\textwidth}
   \includegraphics*[width=\textwidth]{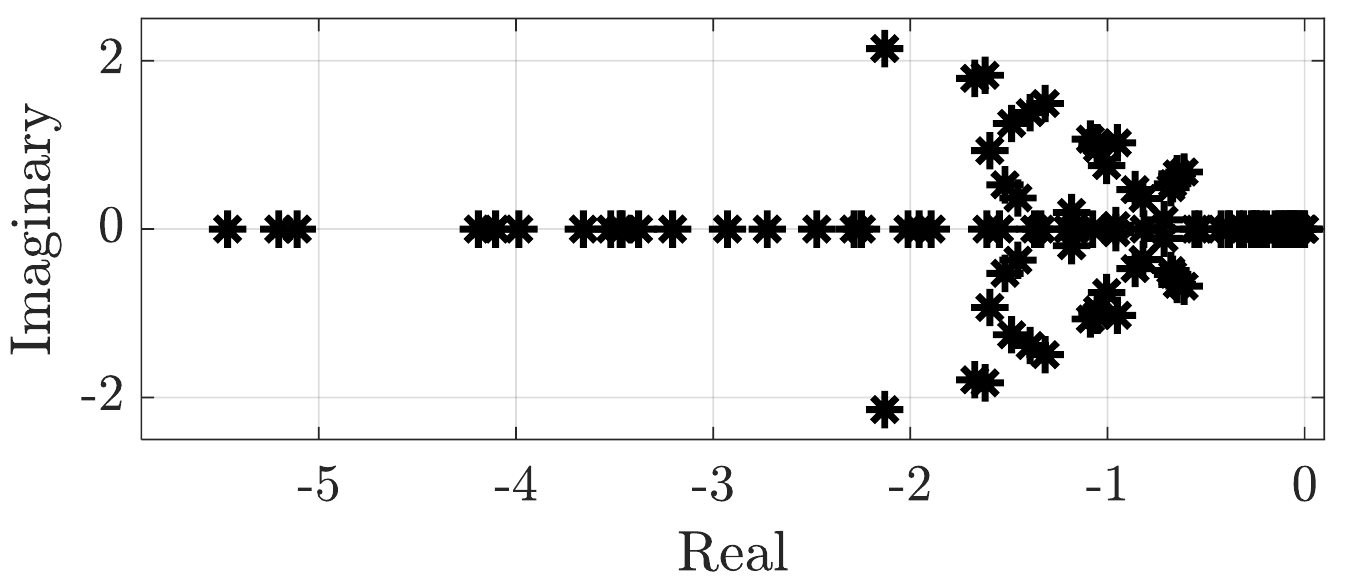}
   \caption{}
   \label{fig: eigens 100-microgrid network}
\end{subfigure}
\caption{Synthetic multi-$\mu$G networks satisfying condition \eqref{eq: condition for stability}. (a) Sparsity pattern of the $50$-microgrid adjacency matrix, black points are ones. (b) $50$-microgrid network. (c) Eigenvalues of the $50$-microgrid network. 
(d) Sparsity pattern of the $100$-microgrid adjacency matrix, black points are ones.
(e) $100$-microgrid network. (f) Eigenvalues of the $100$-microgrid network. }
\label{fig: 50 and 100 mic networks}
\end{figure*}

% \begin{figure*}[]
% \centering
% \begin{subfigure}[b]{0.21\textwidth}
%   \includegraphics*[width=\textwidth]{figures//50_mic_test_case.pdf}
%   \caption{}
%   \label{fig: 50-microgrid network} 
% \end{subfigure}
% \begin{subfigure}[b]{0.26\textwidth}
%   \includegraphics*[width=\textwidth]{figures//50_mic_eigens_edit.pdf}
%   \caption{}
%   \label{fig: eigens 50-microgrid network}
% \end{subfigure}
% %
% %---------------
% \begin{subfigure}[b]{0.21\textwidth}
%   \includegraphics*[width=\textwidth]{figures//100_mic_test_case.pdf}
%   \caption{}
%   \label{fig: 100-microgrid network}
% \end{subfigure}
% \begin{subfigure}[b]{0.26\textwidth}
%   \includegraphics*[width=\textwidth]{figures//100_mic_eigens_edit.pdf}
%   \caption{}
%   \label{fig: eigens 100-microgrid network}
% \end{subfigure}
% \caption{Examples of synthetic multi-$\mu$G networks satisfying condition \eqref{eq: condition for stability}. (a) $50$-microgrid network. (b) Eigenvalues of the $50$-microgrid network. (c) $100$-microgrid network. (b) Eigenvalues of the $100$-microgrid network.}
% \label{fig: 50 and 100 mic networks}
% \end{figure*}
\section{Conclusions} \label{Sec: Conclusions}
This paper proposes new stability certificates for the small-signal stability of multi-$\mu$Gs. 
%We have shown that a quadratic matrix equation links the linearized multi-$\mu$G model to the underling graph of the network. 
%We have focused on the cases where a multi-$\mu$G EP is located within a specific set, and the graph of the network is connected (both physically and electrically). 
%Under such circumstances,
%We empirically showed multi-$\mu$G equilibrium points typically satisfy . 
%belong to a bounded set,
%We have focused on the cases where a multi-$\mu$G EP is located within a specific set
%and such EPs are stable if any of the following two conditions holds (see Theorem \ref{thm: stability properties}): 
In particular, we prove in Theorem \ref{thm: stability properties} that an EP of a multi-$\mu$G system is locally asymptotically stable if either i) the network is lossless; or ii) in a lossy network, a local condition (i.e., condition \eqref{eq: condition for stability}) is satisfied at each microgrid PCC/DER, which roughly speaking requires:
% \begin{align*}
%     \text{Reactive power absorption} + \frac{\text{PCC voltage magnitude}}{\text{Line susceptances}} \le \frac{\text{Damping}^2}{2\cdot\text{Inertia}}. 
% \end{align*}
%--
% \begin{align*}
%     \left( {\begin{array}{*{20}{c}}
%     {\text{Reactive power}}\\
%     {\text{absorption}}
%     \end{array}} \right)
%  %   
%     + \frac{\text{PCC voltage magnitude}}{\text{Line susceptances}} \le \frac{\text{Damping}^2}{2\cdot\text{Inertia}}. 
% \end{align*}
\begin{align*}
    \Big( {\begin{array}{*{20}{c}}
    {\text{Reactive power}}\\
    {\text{absorption}}
    \end{array}} \Big)
    + \frac{\text{Voltage magnitude}}{\text{Line reactance}} \le \frac{\text{Damping}^2}{2\cdot\text{Inertia}}. 
\end{align*}
%
%susceptances
%the sum of two terms is upper bounded by a third term, where the first term is the microgrid absorbing reactive power, the second term is the product of squared terminal voltage magnitude and the absolute value of all line susceptances connected to the microgrid. The third term which bounds the first two terms is proportional to the square of damping and inverse of inertia of each microgrid controller.
This condition sheds new light on the interplay of system stability, network topology, and dynamic parameters. It also provides a fully distributed control scheme that is guaranteed to stabilize the multi-$\mu$G system. The new certificate also reveals an analog of Braess's Paradox in multi-$\mu$G control that adding more lines in the linking grid may worsen system stability, and switching off lines may improve stability margin.
%and provide profound insights into how multi-$\mu$Gs should be designed and operated to be stable.
 The proposed condition in Theorems \ref{thm: stability properties} and \ref{thrm: stability of original net} can improve the situational awareness of system operators by providing a fast stability certificate as well as showing how different corrective actions would make the EP stable.
 %
%
%\textcolor{blue}{
In the literature, several stability criteria are developed  based on various dynamical models, focusing on different aspects of stability. Finding a proper way to compare and merge these criteria and deriving a unified stability criterion will be an interesting direction for future work, and the framework proposed in \cite{2017-vorobev-framework, 2018-Vorobev-Conference-Plug, 2019-Vorobev-Plug} is a promising step towards this direction.
%----------------------------------------------------------------

\appendices
\vspace{-2mm}
\section{Proof of Proposition \ref{prop: geometric nullity of J and J11}} \label{proof of prop: geometric nullity of J and J11}
    \begin{proof}[proof of (i)]
     Let $ (v_1 , v_2) \in \mathbf{ker}(J)$ where $v_1, v_2 \in \mathbb{R}^n$. Then	
	\begin{align} \label{eq: kernel of J}
	\begin{bmatrix}
	0 & I \\
	- M^{-1} L     &     -M^{-1}D \\
	\end{bmatrix}   \begin{bmatrix} v_1 \\v_2  \end{bmatrix}  = 0,
	\end{align}
	which implies that $ v_2 = 0$ and $M^{-1}L v_1 = 0$. Since $M$ is non-singular, $Lv_1=0$, i.e. $v_1 \in \mathbf{ker}(L)$. Therefore, $\mathbf{proj} (\mathbf{ker}(J) ) \subseteq \mathbf{ker}(L)$.
	Conversely, let $v_1 \in \mathbf{ker}(L)$. Set $v_2=0$. Then $ (v_1 , v_2) \in \mathbf{ker}(J)$ as it satisfies \eqref{eq: kernel of J}.  
    \end{proof}
	%-----------------------------------------------------
	\begin{proof}[proof of (ii) and (iii)]
 	 %The geometric multiplicity of an eigenvalue of a matrix is equal to the maximum number of linearly independent eigenvectors associated with the eigenvalue. 
 	 %The eigenvectors of the zero eigenvalue are in the kernel of the matrix.
 	 From part (i) of this proposition, we know that $(v,0)\in \mathbf{ker} (J) \iff v\in \mathbf{ker} (L)$. Therefore, $\{(v_1,0),\cdots,(v_m,0)\}$ is a set of linearly independent eigenvectors in $ \mathbf{ker} (J)$ if and only if $\{v_1,\cdots,v_m\}$ is a set of linearly independent eigenvectors in $ \mathbf{ker} (L)$, i.e., $\dim( \mathbf{ker} (J))=\dim( \mathbf{ker} (L))$. Finally, part (iii) is an immediate consequence of either of the first two parts. 
	\end{proof}
\vspace{-2mm}
\section{Proof of Theorem \ref{thm: stability properties}} \label{app: proof of thm: stability properties}
\begin{proof}[proof of (\ref{thrm-a})]
      This is an immediate consequence of Propositions \ref{prop: geometric nullity of J and J11} and \ref{prop: quasi-M}.  %   since the underlying undirected graph of the system is connected and the set of zero cuts is empty, $L$ has one zero eigenvalue 
     %   $L$ for a lossless network is positive semidefinite with one zero eigenvalue. According to Lemma \ref{lemma: nullity of J and J11}, the null space of $L$ is one dimensional defined by the eigenvector $v_0:=[1;...;1]$. 
    %   The null space of $J$ is also one dimensional defined by the associated eigenvector  $[v_0; 0]$ of $J$ by Proposition \ref{Prop: singularity of swing equtions}. This proves the theorem.
\end{proof}      
%--------------
% \begin{proof}[proof (b)]
% \textcolor{red}{Modify this definition also}
% \end{proof}
%--------------
\begin{proof}[proof of (\ref{thrm-b})]
See \cite[Proof of Theorem 1]{2020-fast-certificate}.
\end{proof}
%--------------
\begin{proof}[proof of (\ref{thrm-c})]
The result holds for real nonzero eigenvalues of $J$, as shown in the previous part.
Now let $\lambda\in\mathbb{C},  \lambda\in\sigma(J)$, then according to Lemma \ref{lemma: relation between ev J and ev J11}, $\exists v\in\mathbb{C}^n,v\ne0$ such that
        \begin{align} \label{eq: pencil singularity theorem proof}
            \left( L + \lambda D + \lambda^2 M \right) v = 0.
        \end{align}
It is always possible to normalize $v$ such that $\max_{i\in\mathcal{N}} |v_i| = 1$.
Here and in the rest of this proof, if $x\in\mathbb{C}$, then $|x|$ denotes the modulus of $x$.
Let $k:=  \mathrm{argmax}_{i\in\mathcal{N}} |v_i|$, and spell out the $k$-th row of \eqref{eq: pencil singularity theorem proof}:
%         \begin{align} \label{eq: pencil singularity theorem proof k-th row}
%             \sum_{i=1}^n L_{ki}v_i + \lambda d_k v_k + \lambda^2 m_k v_k   = 0,
%         \end{align}
% which can be rewritten as
        \begin{align} \label{eq: pencil singularity theorem proof k-th row rewrite}
            L_{kk}v_k + \lambda d_k v_k + \lambda^2 m_kv_k   = -\sum_{i=1, i\ne k}^n L_{ki}v_i.
        \end{align}
Using the triangle inequality, we have
\begin{align*}
  \bigl\lvert-\sum_{i=1, i\ne k}^n L_{ki}v_i \bigr\rvert \le  
  \sum_{i=1, i\ne k}^n \bigl\lvert L_{ki}\bigr\rvert \bigl\lvert v_i \bigr\rvert \le 
  \sum_{i=1, i\ne k}^n \bigl\lvert L_{ki}\bigr\rvert.
\end{align*}
Define $\mathcal{R}:=\sum_{i=1, i\ne k}^n \bigl\lvert L_{ki}\bigr\rvert$.
% \begin{align*}
%     \mathcal{R}:=\sum_{i=1, i\ne k}^n \bigl\lvert L_{ki}\bigr\rvert.
% \end{align*}
Now assume that $\lambda=\alpha+j\beta$ with $\alpha\ge0, \beta \ne 0$ is a nonzero eigenvalue of $J$, and let us lead this assumption to a contradiction. Recall $|v_k|=\|v\|_\infty=1$. Equation \eqref{eq: pencil singularity theorem proof k-th row rewrite} implies that
%\label{eq: pencil singularity theorem proof k-th row rewrite abs}
\begin{align*} 
 \mathcal{R}^2 \ge & \bigl\lvert  L_{kk}v_k + \lambda d_k v_k + \lambda^2 m_kv_k \bigr\rvert^2 
 =  \bigl\lvert  L_{kk} + \lambda d_k  + \lambda^2 m_k \bigr\rvert^2  \\
 %=&   \bigl\lvert  L_{kk} + \alpha d_k + (\alpha^2 -\beta^2) m_k +j(2\alpha\beta m_k +\beta d_k)      \bigr\rvert^2 \\
 %=& \left( L_{kk} + \alpha d_k + (\alpha^2 -\beta^2) m_k \right)^2 + \left( 2\alpha\beta m_k +\beta d_k \right)^2 \\
 =& L_{kk}^2 + (\alpha d_k + (\alpha^2 -\beta^2) m_k)^2 + 2 L_{kk} (\alpha d_k +\alpha^2 m_k) \\
  & - 2 L_{kk}\beta^2 m_k  + 4\alpha^2\beta^2 m_k^2 +\beta^2 d_k^2 + 4\alpha\beta^2 m_k d_k.
\end{align*}
Recall that if $(\delta^{*},\omega^{*})\in \Omega$, matrix $L$ has zero row sum, i.e., $\mathcal{R} = L_{kk}$. By cancelling $\mathcal{R}^2$ and $L_{kk}^2$ terms and moving $2L_{kk}\beta^2 m_k$ and $\beta^2 d_k^2$ to the left-hand side, we arrive at
\begin{align} \label{eq: ineq in proof of lossy}
\notag  \beta^2 (2L_{kk} m_k - d_k^2) \ge &  (\alpha d_k + (\alpha^2 -\beta^2) m_k)^2 \\
\notag  &   + 2 L_{kk} (\alpha d_k +\alpha^2 m_k) \\
  &   + 4\alpha^2\beta^2 m_k^2 + 4\alpha\beta^2 m_k d_k.
\end{align}
Now, note that the outgoing reactive power flow at PCC $k$ is 
\begin{align} \label{eq: reactive power flow}
   \notag Q_k &= - \sum \limits_{i=1}^n { V_k V_i Y_{ki} \sin \left( {\theta _{ki} - {\delta _k^*} + {\delta _i^*}} \right) } \\
  \notag  &= - V_k^2 B_{kk} - \sum \limits_{i=1, i\ne k}^n { V_k V_i Y_{ki} \sin \left( {\theta _{ki} - {\delta _k^*} + {\delta _i^*}} \right) } \\
    &= - V_k^2 B_{kk} - L_{kk} ,
\end{align}
where $B_{kk}=Y_{kk}\sin(\theta_{kk})$. Therefore condition \eqref{eq: condition for stability} implies that $(2L_{kk} m_k - d_k^2)\le0$, 
thus the left-hand side of the inequality \eqref{eq: ineq in proof of lossy} is nonpositive. If $\alpha\ge0$ and $\beta\ne0$, the right-hand side of \eqref{eq: ineq in proof of lossy} will be positive, which is the desired contradiction.
%Note that as the power network is connected and the set of zero cuts is empty, we have $L_{kk}>0$.
%The idea used in this part of the proof was partially inspired by Skar \cite{1980-Skar-stability-thesis}.
%
According to Proposition \ref{prop: geometric nullity of J and J11}, the simple zero eigenvalue of the Jacobian matrix $J$ stems from the translational invariance of the flow function \eqref{eq: flow function}. As mentioned earlier, we can eliminate this eigenvalue by choosing a reference bus and refer all other bus angles to it. Therefore, the set of EPs $\{ \delta^* + \alpha \mathbf{1}: \alpha \in \mathbb{R} \}$ will collapse into one EP. Such an EP will be asymptotically stable.  
\end{proof}
%--------------
\begin{proof}[proof of (\ref{thrm-d})]
    Let $\lambda \in \sigma(J)$, then according to Lemma \ref{lemma: relation between ev J and ev J11}, $\exists v\ne0$ such that $( L + \lambda D + \lambda^2 M ) v = 0.$
    % \label{eq: singularity of pencil}
        % \begin{align*} 
        %     \left( L + \lambda D + \lambda^2 M \right) v = 0.
        % \end{align*}
    According to part (\ref{thrm-b}) of this proof, if the eigenvalue $\lambda$ is a real number, the desired result holds. We complete the proof in two steps:
% \begin{enumerate}[(i)]
%     \item 

    \noindent \textbf{Step 1:} First, we prove that the eigenvalues of $J$ cannot be purely imaginary. 
% 	\begin{align}
% 	\begin{bmatrix}
% 	0 & I \\
% 	-0.5 H^{-1}  L      &     -0.5 D H^{-1} \\
% 	\end{bmatrix}   \begin{bmatrix} v_1 \\v_2  \end{bmatrix}  = \lambda   \begin{bmatrix} v_1 \\v_2  \end{bmatrix},
% 	\end{align}
% 	where $[v_1, v_2]$ is an associated eigenvector. Note that 
% 	$ v_2 = \lambda v_1$, thus
% 	\begin{align}
% 	& \left[     \frac{1}{2} H^{-1}  L  + \lambda ( \frac{1}{2}H^{-1} D  + \lambda I  )       \right]   v_1 = 0 \\
% 	& \Rightarrow \;  \frac{1}{2}H^{-1}\left(  L  + \lambda D + 2 \lambda^2 H \right) v_1 = 0 \\
% 	& \Rightarrow \; \left(  L  + \lambda D + 2 \lambda^2 H \right) v_1 = 0. \label{eq:lambda}
% 	\end{align}  
	Suppose $\lambda  = j \beta \in \sigma(J)$ for some nonzero real $\beta$. Let $v  =  x + j y$, then $((L - \beta^2 M) + j\beta D) (x+jy) = 0$,
% 	\begin{align}
% 	((L - \beta^2 M) + j\beta D) (x+jy) = 0,
% 	\end{align}
	which can be equivalently written as
	\begin{align}
	\begin{bmatrix}
	L - \beta^2 M & -\beta D \\
	\beta D & L - \beta^2 M
	\end{bmatrix} \begin{bmatrix}
	x\\y
	\end{bmatrix} = \begin{bmatrix}
	0\\0
	\end{bmatrix}.
	\end{align}
	Define the matrix 
	\begin{align} \label{eq:M}
	H(\beta) := \begin{bmatrix}
	\beta D &  L  - \beta^2 M\\
	 L  - \beta^2 M & -\beta D
	\end{bmatrix}. 
	\end{align}
	For a lossless network, $L \succeq 0$ 
	%$L$ is a Laplacian matrix, therefore positive semidefinite  with exactly one zero eigenvalue
	(see Section \ref{subsec: digraph}). Thus, $H(\beta)$ is a symmetric matrix.
	%, also notice that $H(\beta)$ cannot be positive semidefinite due to the diagonal block $-\beta D$. 
	Since $D\succ 0$, the determinant of $H(\beta)$ can be expressed using Schur complement as
	\begin{align*}
	 \mathbf{det}(H(\beta)) & =  \mathbf{det} (-\beta D) \mathbf{det} (\beta D \\
	  &+ \beta^{-1} ( L  - \beta^2 M )D^{-1}(  L  - \beta^2 M)).
	\end{align*}
	%So we only need to consider the nonsingularity of the Schur complement, which can be expressed as
	Define the following matrices for the convenience of analysis: 
	%$	A(\beta) :=  L  - \beta^2 M$, $B(\beta) := D^{-1/2}A(\beta)D^{-1/2}$,
	%	$E(\beta) := I + \beta^{-2} B(\beta)^2$.
 	\begin{align*}
 		A(\beta) &:=  L  - \beta^2 M, \\
 		B(\beta) &:= D^{-1/2}A(\beta)D^{-1/2},\\
 		E(\beta) &:= I + \beta^{-2} B(\beta)^2.
 	\end{align*}
% 	}
	The inner matrix of the Schur complement can be written as
	% \label{eq:schur}
	%\begin{subequations}
		\begin{align*}
		& \beta D + \beta^{-1} ( L  - \beta^2 M )D^{-1}(  L  - \beta^2 M) \\ 
	 & =  \beta D^{1/2} (I + \beta^{-2} D^{-1/2}A(\beta)D^{-1}A(\beta)D^{-1/2})D^{1/2} \\
	 & =  \beta D^{1/2} (I + \beta^{-2} B(\beta)^2)D^{1/2} = \beta D^{1/2} E(\beta) D^{1/2}.
		\end{align*}
	%\end{subequations}
	Notice that {$E(\beta)$ and $B(\beta)$ have the same eigenvectors and the eigenvalues of $E(\beta)$ and $B(\beta)$ have a one-to-one correspondence: $\mu$ is an eigenvalue of $B(\beta)$ if and only if $1 + \beta^{-2} \mu^2$ is an eigenvalue of $E(\beta)$}. Indeed, we have 
	 $E(\beta)v = v + \beta^{-2}B(\beta)^2v = v + \beta^{-2}\mu^2 v = (1+\beta^{-2}\mu^2)v$ for any eigenvector $v$ of $B(\beta)$ with eigenvalue $\mu$. 
	Since $B(\beta)$ is symmetric, $\mu$ is a real number. Hence, $E(\beta)$ is positive definite (because $1+\beta^{-2}\mu^2>0$), therefore $H(\beta)$ is nonsingular for any real nonzero $\beta$. Then the eigenvector $v=x+jy$ is zero which is a contradiction. This proves that $J$ has no eigenvalue on the punctured imaginary axis.\\
	%Note that $J$ could have complex eigenvalues as shown in the single-machine to infinite-bus case.
	%----------------------------------------------------
    % 	\item Secondly, if the eigenvalue $\lambda$ is a real number, then since $L$, $D$, and $M$ are all positive semidefinite, the only way to make $ L  + \lambda D + \lambda^2 M$ singular is to have $\lambda\le 0$. So all the real nonzero eigenvalues of $J$ are negative.
	% Note that $v_1 \ne 0$ since otherwise $v_2 = \lambda v_1$ forces the eigenvector to be zero, which is a contradiction.
	%\item
	\textbf{Step 2:}
	 Second, for a complex eigenvalue $\alpha + j\beta$ of $J$ with $\alpha\ne 0, \beta\ne 0$, by setting $v  =  x + j y$, the pencil singularity equation becomes
	\begin{align*}
	( L  + (\alpha + j\beta) D + (\alpha^2 - \beta^2 + 2\alpha\beta j)M)(x+jy) = 0.
	\end{align*}
	Similar to Step 1 of the proof, define the matrix $H(\alpha,\beta)$ as
	\begin{align*}
	\begin{bmatrix}
	 L  + \alpha D + (\alpha^2-\beta^2) M & -\beta(D + 2\alpha M) \\
	\beta(D+2\alpha M) &  L  + \alpha D + (\alpha^2-\beta^2) M
	\end{bmatrix}.
	\end{align*}
	We only need to consider two cases, namely i) $\alpha > 0, \beta > 0$ or ii) $\alpha <0, \beta > 0$. For the first case, $\beta(D+{2}\alpha M)$ is invertible and positive definite, therefore, we only need to look at the invertibility of the Schur complement
	\begin{align*}
	& S(\alpha,\beta) + T(\alpha,\beta) S^{-1}(\alpha,\beta) T(\alpha,\beta),
	\end{align*}
	where $S(\alpha,\beta):= \beta(D+{2}\alpha M)$ and $T(\alpha,\beta):= L  + \alpha D + (\alpha^2-\beta^2) M$.
% 	\begin{align*}
% 	    & S(\alpha,\beta):= \beta(D+{2}\alpha M), \\
% 	    & T(\alpha,\beta):= L  + \alpha D + (\alpha^2-\beta^2) M.
% 	\end{align*}
	Using the same manipulation as in Step 1 of the proof, we can see that the Schur complement is always invertible for any $\alpha >0, \beta>0$. This implies the eigenvector $v$ is 0, which is a contradiction. Therefore, the first case is not possible. So any complex nonzero eigenvalue of $J$ has negative real part. 
%\end{enumerate}
%\vspace{-2mm}
\end{proof}
%------------------------------------------------------------------------
%-----------------------------------------------------------------------
%----------------------------------------------------------------------
% ADDED In R1
\vspace{5mm}
%\textcolor{blue}{
\section{Proof of Lemma \ref{lem: assumption closed Kron}}\label{sec: proof of lem: assumption closed Kron}
%
%The proof of Lemma \ref{lem: schur-complement-close} can be found for example in 
% Let us first state the following lemma whose proof can be found in \cite{2012-dorfler-kron}.
% \begin{lemma} \label{lem: schur-complement-close}
%     The following two classes of matrices are invariant under Kron reduction: i) matrices with zero row sum; ii) symmetric matrices. In other words, zero row summation and symmetry of matrices remain unchanged after Kron reduction.
% \end{lemma}
%
\begin{proof}
Consider the nodal admittance matrix $Y\in\mathbb{C}^{n \times n}$ which satisfies Assumptions \ref{as: y_matrix sign} and \ref{as: y_matrix bound}. Let $Y$ induce a network $\mathcal{G} = (\mathcal{N}, \mathcal{E})$ with the set of active nodes $\alpha \subset \mathcal{N}$ and passive nodes $\beta  = \mathcal{N} \setminus \alpha$. 
%After properly labeling the nodes, we can have $\beta = \{n-|\beta| + 1, \cdots , n\}$. 
%The purpose of Kron reduction here is to remove the set of passive nodes $\beta$,
%Let the distribution network $\mathcal{G} = (\mathcal{N}, \mathcal{E})$ have the nodal admittance matrix $Y\in\mathbb{C}^{n \times n}$.
%
%
%
According to Definition \ref{def: Kron}, the Kron reduced matrix after removing node $k_0\in\beta$ is
%and this can be accomplished by constructing a sequence of matrices
%$\{Y^{(\ell)}\}_{\ell =1}^{|\beta|}$, where
$Y^{r}\in\mathbb{C}^{(n-1)\times(n-1)}$ defined as
\begin{align} \label{eq: iterative Kron}
	Y^{r}_{ik} = Y_{ik} - {Y_{ik_0}Y_{k_0 k}} /{Y_{k_0 k_0}}, \: \forall i,k\ne k_0
\end{align}
% where $i, k \in\{ 1, \cdots, n-\ell \}$,
% $Y^{(0)} = Y$, $Y^{(|\beta|)} = Y^{r}$, and $m_\ell = n - \ell +1$. Recall $Y^{r}$ is the Kron reduced admittance matrix defined in Definition \ref{def: Kron}.
% Observe that the matrix sequence $\{Y^{(\ell)}\}_{\ell =1}^{|\beta|}$ is well-defined.
%
% In order to prove $Y^r \in \mathcal{Y}\cap\mathcal{Y}'$, it suffices to prove that for each $\ell \in \{1,\cdots,|\beta|\}$ matrix $Y^{(\ell)}$ satisfies Assumptions \ref{as: y_matrix sign}.
First, we prove that $Y^r$ satisfies Assumption \ref{as: y_matrix sign}.
Recall that the following two classes of matrices are invariant under Kron reduction \cite{2012-dorfler-kron}: i) matrices with zero row sum; ii) symmetric matrices. In other words, $Y^{r}$ is a symmetric matrix with zero row sum.
Hence, we can restrict our analysis to off-diagonal entries, and aim to prove that $Y^r=G^r+jB^r$ satisfies $G_{ik}^{r} \le0, B_{ik}^{r}\ge0,$ for all $i\ne k$.
%
%and according to Lemma \ref{lem: schur-complement-close}, each $Y^{(\ell)}$ is a nodal admittance matrix of a well-defined network.
%
%
%
Consider $Y_{ik} = G_{ik} + j B_{ik}$ with $G_{ik}\le0$ and $B_{ik}\ge0$
and note that
%we have dropped the superscript $(\ell-1)$ on $G_{ik}$ and $B_{ik}$ for notational simplicity.
for off-diagonal entries $Y^{r}_{ik}, i\ne k$, we have 
% \begin{align}
% 	Y^{(\ell)}_{ik} &= Y^{{(\ell-1)}}_{ik} - {Y^{(\ell-1)}_{im_\ell}Y^{(\ell-1)}_{m_\ell k}} /{Y^{(\ell-1)}_{m_\ell m_\ell}} \\
% 	&= G_{ik} + j B_{ik} - (G_{im_\ell} + j B_{im_\ell})(G_{m_\ell k} + j B_{m_\ell k}) / (G_{m_\ell m_\ell } + j B_{m_\ell m_\ell}) \\
% 	& = G_{ik} + j B_{ik} - 
% 	(	(G_{im_\ell}G_{m_\ell k} - B_{im_\ell}B_{m_\ell k})
% 	+j(G_{im_\ell}B_{m_\ell k} + B_{im_\ell}G_{m_\ell k})	) / (G_{m_\ell m_\ell } + j B_{m_\ell m_\ell})\\
% 	%
% 	& = G_{ik} + j B_{ik} - 
% 	(	(G_{im_\ell}G_{m_\ell k} - B_{im_\ell}B_{m_\ell k})
% 	+j(G_{im_\ell}B_{m_\ell k} + B_{im_\ell}G_{m_\ell k})	) (G_{m_\ell m_\ell } - j B_{m_\ell m_\ell}) / (G_{m_\ell m_\ell }^2 +  B_{m_\ell m_\ell}^2) \\
% 	%
% 	%
% 	& = G_{m_\ell m_\ell }(G_{im_\ell}G_{m_\ell k} - B_{im_\ell}B_{m_\ell k})
% 	+B_{m_\ell m_\ell}(G_{im_\ell}B_{m_\ell k} + B_{im_\ell}G_{m_\ell k}) \\
% 	%
% 	&+j(
% 	G_{m_\ell m_\ell }(G_{im_\ell}B_{m_\ell k} + B_{im_\ell}G_{m_\ell k})
% 	-B_{m_\ell m_\ell}(G_{im_\ell}G_{m_\ell k} - B_{im_\ell}B_{m_\ell k})
% 	)
% 	%
% \end{align}
\begin{align*}
	Y_{ik} - &Y^{r}_{ik} =   {Y_{ik_0}Y_{k_0 k}} /{Y_{k_0 k_0}} \\
	& =  (G_{ik_0} + j B_{ik_0})(G_{k_0 k} + j B_{k_0 k}) / (G_{k_0 k_0 } + j B_{k_0 k_0}) \\
	%=& 
	%(	(G_{ik_0}G_{m_\ell k} - B_{im_\ell}B_{m_\ell k}) \\
	%&+j(G_{im_\ell}B_{m_\ell k} + B_{im_\ell}G_{m_\ell k})	) / (G_{m_\ell m_\ell } + j B_{m_\ell m_\ell})\\
	%
	& =  
	(	(G_{ik_0}G_{k_0 k} - B_{ik_0}B_{k_0 k}) \\
	&\hspace{3mm} +j(G_{ik_0}B_{k_0 k} + B_{ik_0}G_{k_0 k})	) (G_{k_0 k_0 } - j B_{k_0 k_0}) /\eta,
	%
	%
	%=& ( G_{m_\ell m_\ell }(G_{im_\ell}G_{m_\ell k} - B_{im_\ell}B_{m_\ell k}) \\
	%&+B_{m_\ell m_\ell}(G_{im_\ell}B_{m_\ell k} + B_{im_\ell}G_{m_\ell k}) )/\eta \\
	%
	%& +j(
	%G_{m_\ell m_\ell }(G_{im_\ell}B_{m_\ell k} + B_{im_\ell}G_{m_\ell k}) \\
	%&-B_{m_\ell m_\ell}(G_{im_\ell}G_{m_\ell k} - B_{im_\ell}B_{m_\ell k})
	%)/\eta
	%
\end{align*}
where $\eta = G_{k_0 k_0 }^2 +  B_{k_0 k_0}^2$. 
Observe that
\begin{align*}
\Im(
{Y_{ik_0}  Y_{k_0 k}}  /{Y_{k_0 k_0}}
) \eta 
    =& G_{k_0 k_0 }G_{ik_0}B_{k_0 k} + G_{k_0 k_0 }B_{ik_0}G_{k_0 k} \\
	 & -B_{k_0 k_0}(G_{ik_0}G_{k_0 k} - B_{ik_0}B_{k_0 k}) 
	\le 0,
\end{align*}
where the inequality holds 
because under Assumptions \ref{as: y_matrix sign} and \ref{as: y_matrix bound}, we have
\begin{align*}
   &G_{k_0 k_0 }G_{ik_0}B_{k_0 k} \le 0,
   G_{k_0 k_0 }B_{ik_0}G_{k_0 k} \le 0,\\
   &G_{ik_0}G_{k_0 k} - B_{ik_0}B_{k_0 k} \le 0,
   -B_{k_0 k_0} \ge 0. 
\end{align*}
This proves that $B_{ik}^{r}\ge0,$ for all $i\ne k$.
Also observe that
\begin{align*}
  \Re(
  {Y_{ik_0}Y_{k_0 k}} & /{Y_{k_0 k_0}}
  ) \eta\\
    %& =
    %G_{m_\ell m_\ell }(G_{im_\ell}G_{m_\ell k} - B_{im_\ell}B_{m_\ell k}) \\
	%&+B_{m_\ell m_\ell}(G_{im_\ell}B_{m_\ell k} + B_{im_\ell}G_{m_\ell k}) \\
	&= G_{k_0 k_0}G_{ik_0}G_{k_0 k} - G_{k_0 k_0}B_{ik_0}B_{k_0 k}\\
	&\hspace{4mm} +B_{k_0 k_0}G_{ik_0}B_{k_0 k} + B_{k_0 k_0}B_{ik_0}G_{k_0 k}\\
	&\ge  G_{k_0 k_0}G_{ik_0}G_{k_0 k}   - \nu_{\max}^2 G_{k_0 k_0}G_{ik_0}G_{k_0 k} \\
	& \hspace{4mm}+ \nu_{\min}^2  G_{k_0 k_0}G_{ik_0}G_{k_0 k} + \nu_{\min}^2  G_{k_0 k_0}G_{ik_0}G_{k_0 k} \\
	&=(1+2\nu_{\min}^2-\nu_{\max}^2) G_{k_0 k_0}G_{ik_0}G_{k_0 k} \ge 0,
	% This is the old assumption
    % 	&\ge -\nu_{\min}^3 B_{k_0 m_\ell} B_{im_\ell} B_{m_\ell k} + \nu_{\max} B_{m_\ell m_\ell}B_{im_\ell}B_{m_\ell k}\\
    % 	&- \nu_{\min} B_{m_\ell m_\ell} B_{im_\ell} B_{m_\ell k} - \nu_{\min} B_{m_\ell m_\ell}B_{im_\ell}B_{m_\ell k} \\
    % 	&=  -(\nu_{\min}^3 + 2\nu_{\min} - \nu_{\max} ) B_{m_\ell m_\ell} B_{im_\ell} B_{m_\ell k} \ge 0,
	%
    % 	&= G_{m_\ell m_\ell} G_{im_\ell} G_{m_\ell k} - \nu^2 G_{m_\ell m_\ell} G_{im_\ell} G_{m_\ell k} 
    % 	+\nu^2 G_{m_\ell m_\ell}G_{im_\ell}G_{m_\ell k} + \nu^2 G_{m_\ell m_\ell}G_{im_\ell}G_{m_\ell k} \\
    % 	&= (1+\nu^2)G_{m_\ell m_\ell}G_{im_\ell}G_{m_\ell k} \ge 0
\end{align*}
where the inequality holds
because under Assumptions \ref{as: y_matrix sign} and \ref{as: y_matrix bound}, we have $1+2\nu_{\min}^2-\nu_{\max}^2 \ge 0$  and $G_{k_0 k_0} G_{ik_0} G_{k_0 k} \ge 0$.
% \begin{align*}
%     1+2\nu_{\min}^2-\nu_{\max}^2 \ge 0, \text{ and }
%     G_{k_0 k_0} G_{ik_0} G_{k_0 k} \ge 0.
% \end{align*}
This shows that $G_{ik}^{r} \le0$ for all $i\ne k$, and completes the first part of proof.
%------------------------------------
Next, we prove that $B^r_{kk}\ge B_{kk}$ for all $k\ne k_0$. Observe that
\begin{align*}
\Im(Y^{r}_{kk}  - Y_{kk}) {\eta}
%
%& = - \frac{  -(G_{i m_\ell}^2 - B_{i m_\ell}^2)B_{m_\ell m_\ell}+ 2G_{m_\ell m_\ell}G_{i m_\ell}B_{i m_\ell} } {G_{m_\ell m_\ell }^2 +  B_{m_\ell m_\ell}^2} \\
 = {(G_{k k_0}^2 - B_{k k_0}^2)B_{k_0 k_0}}  
- {2G_{k_0 k_0}G_{k k_0}B_{k k_0}}.
\end{align*}
According to Assumption \ref{as: y_matrix sign}, we have $G_{k_0 k_0}\ge0, B_{k_0 k_0} \le 0$, $G_{kk_0}\le0, B_{kk_0} \ge 0, \forall k\ne k_0$. Assumption \ref{as: y_matrix bound} says 
$|G_{k k_0}| \le |B_{k k_0}|$.
Hence ${  (G_{kk_0}^2 - B_{k k_0 }^2)B_{k_0 k_0}   }  \ge 0$ and $-{  2G_{k_0 k_0 }G_{k k_0 }B_{k k_0}   }  \ge 0.$
% \begin{align*}
% {  (G_{kk_0}^2 - B_{k k_0 }^2)B_{k_0 k_0}   }  \ge 0, -{  2G_{k_0 k_0 }G_{k k_0 }B_{k k_0}   }  \ge 0.
% \end{align*}
This implies that $B^{r}_{kk} \ge B_{kk}$, and completes the proof.
\end{proof}
\vspace{5mm}
\section{Proof of Theorem \ref{thrm: stability of original net}}
\label{sec: proof of thrm: stability of original net}

\begin{proof}
%\textcolor{blue}{
Let the nodal admittance matrix of $\mathcal{G}^d$ be $Y\in\mathbb{C}^{n \times n}$ which satisfies Assumptions \ref{as: y_matrix sign} and \ref{as: y_matrix bound}. Suppose $\mathcal{G}^d$ has the set of active nodes $\alpha \subset \mathcal{N}^d$ and passive nodes $\beta  = \mathcal{N}^d \setminus \alpha$. 
After properly labeling the nodes, we can have $\beta = \{n-|\beta| + 1, \cdots , n\}$. 
In order to get the admittance matrix $Y^r$ of the Kron reduced network $\mathcal{G}^r$, we need to remove the set of passive nodes $\beta$ according to Definition \ref{def: Kron},
%Let the distribution network $\mathcal{G} = (\mathcal{N}, \mathcal{E})$ have the nodal admittance matrix $Y\in\mathbb{C}^{n \times n}$.
%the Kron reduced admittance matrix defined in 
%
%
%
and this can be accomplished by constructing a sequence of matrices $\{Y^{(\ell)}\}_{\ell =1}^{|\beta|}$, where $Y^{(\ell)}\in\mathbb{C}^{(n-\ell)\times(n-\ell)}$ is defined as
\begin{align} \label{eq: iterative Kron}
	Y^{(\ell)}_{ik} = Y^{{(\ell-1)}}_{ik} - {Y^{(\ell-1)}_{im_\ell}Y^{(\ell-1)}_{m_\ell k}} /{Y^{(\ell-1)}_{m_\ell m_\ell}},
\end{align}
where $i, k \in\{ 1, \cdots, n-\ell \}$,
$Y^{(0)} = Y$, $Y^{(|\beta|)} = Y^{r}$, and $m_\ell = n - \ell +1$. 
%Recall $Y^{r}$ is the Kron reduced admittance matrix defined in Definition \ref{def: Kron}.
Observe that the matrix sequence $\{Y^{(\ell)}\}_{\ell =1}^{|\beta|}$ is well-defined.
Now, according to Lemma \ref{lem: assumption closed Kron}, for each $\ell \in \{1,\cdots,|\beta|\}$ matrix $Y^{(\ell)}$ satisfies Assumptions \ref{as: y_matrix sign}. Hence, $Y^r$ satisfies both Assumptions \ref{as: y_matrix sign} and \ref{as: y_matrix bound}.
%}
%it suffices to prove that.
% According to Lemma \ref{lem: schur-complement-close}, each $Y^{(\ell)}$ is a symmetric matrix with zero row sum.
%
% Hence, we can restrict our analysis to off-diagonal entries, and aim to prove that each $Y^{(\ell)}=G^{(\ell)}+jB^{(\ell)}$ satisfies $G_{ik}^{(\ell)} \le0, B_{ik}^{(\ell)}\ge0,$ for all $i\ne k$.
%     {
% 	Let the power network $\mathcal{G} = (\mathcal{N}, \mathcal{E})$ have the nodal admittance matrix $Y\in\mathbb{C}^{n \times n}$ with the set of passive nodes $\beta \subset \mathcal{N}$. 
% 	After properly labeling the nodes, we can have $\beta = \{n-|\beta| + 1, \cdots , n\}$.
% 	The Kron reduction of this network can be accomplished by constructing a sequence of matrices $\{Y^{(\ell)}\}_{\ell =1}^{|\beta|}$, where $Y^{(\ell)}\in\mathbb{C}^{(n-\ell)\times(n-\ell)}$ is defined as
% 	\begin{align}
% 		Y^{(\ell)}_{ik} = Y^{{(\ell-1)}}_{ik} - {Y^{(\ell-1)}_{im_\ell}Y^{(\ell-1)}_{m_\ell k}} /{Y^{(\ell-1)}_{m_\ell m_\ell}}, && i, k \in\{ 1, \cdots, n-\ell \}
% 	\end{align}
%     where $Y^{(0)} = Y$, $Y^{(|\beta|)} = Y^{r}$, and $m_\ell = n - \ell +1$. The matrix sequence $\{Y^{(\ell)}\}_{\ell =1}^{|\beta|}$ is well-defined.
%     %
%     %
%     Consider $Y_{ik}^{(\ell - 1)} = G_{ik} + j B_{ik}$ and note that
%     we have dropped the superscript $(\ell-1)$ on $G_{ik}$ and $B_{ik}$ for notational simplicity.
%     }
%

%\textcolor{blue}{
Next, Let $V\in\mathbb{C}^n$ and $S\in\mathbb{C}^n$ be the vector of nodal voltages and power injections of network $\mathcal{G}^d$, respectively.
%Let $\alpha\subset\{1,...,n\}$ denote the set of active nodes. 
It can be shown that if the vector of nodal voltages of the reduced network $\mathcal{G}^r$ is $V[\alpha]$, then the vector of power injections in the reduced network is $S[\alpha]$.
Hence, if the voltage magnitudes in the original and Kron-reduced networks are equal, then the reactive power $Q_i$ at active nodes in the two networks are equal.
Moreover, Lemma \ref{lem: assumption closed Kron} asserts that  $B^{(\ell)}_{ii} \ge B_{ii}^{(\ell-1)}$, for all $i \in\{ 1, \cdots, n-\ell \}$.
Since this inequality holds for all $\ell \in \{1,\cdots, |\beta| \}$,
by induction, we conclude that
$
	B_{ii}^{r} \ge B_{ii},  \forall i \in \{1, \cdots, n-|\beta| \}
$
where $B_{ii}^{r}$ and $B_{ii}$ are the $i$th diagonal entries of the Kron-reduced and original admittance matrices, respectively.
Note that
$
	-Q_i -B_{ii}^{r}V_i^2 \le -Q_i -B_{ii}V_i^2.
$
%
%Now, 
%
Therefore, if $-Q_i -B_{ii}V_i^2 \le {d_i^2} / {2m_i}$ holds for active nodes in the original network, then $-Q_i -B_{ii}^{r}V_i^2 \le \frac{d_i^2}{2m_i}$ also holds and according to Theorem \ref{thm: stability properties}, the stability of the system is guaranteed.
%}
\end{proof}

\newpage
\bibliographystyle{IEEEtran}
\bibliography{References}

\end{document}